\documentclass[sigconf,nonacm]{acmart}

\AtBeginMaketitle{%

\def\EDBTISSN{xxxx-xxxx}
\def\EDBTISBN{xxx-x-xxxxx-xxx-x}
\setcopyright{none}
\copyrightyear{\copyright\ 2027 Copyright held by the owner/author(s).
  Published on xxxxx under ISBN \EDBTISBN, series ISSN \EDBTISSN.
  Distribution of this paper is permitted under the terms of the
  Creative Commons license CC-by-nc-nd 4.0}
\acmYear{2027}
\acmDOI{}
\acmISBN{}

\acmConference[EDBT '27]{Extending Database Technology}{6-9 April 2027}{Lille (France)}

\settopmatter{printacmref=false, printccs=false, printfolios=false}

\newsavebox{\ximagebox}
\newlength{\ximageheight}
\newsavebox{\xglyphbox}
\newlength{\xglyphheight}
\newcommand{\xbox}[1]%
  {\savebox{\ximagebox}{#1}%
  \settoheight{\ximageheight}{\usebox{\ximagebox}}%
  \savebox{\xglyphbox}{\color{white}\char32}%
  \settoheight{\xglyphheight}{\usebox{\xglyphbox}}%
  \raisebox{\ximageheight}[0pt][0pt]{\raisebox{-\xglyphheight}[0pt][0pt]{%
    \makebox[0pt][l]{\usebox{\xglyphbox}}}}%
    \usebox{\ximagebox}%
    \raisebox{0pt}[0pt][0pt]{\makebox[0pt][r]{\usebox{\xglyphbox}}}}

\newsavebox{\LogoBox}
\sbox{\LogoBox}{}

\fancypagestyle{firstpagestyle}{%
  \fancyhf{}%
  \fancyfoot{}%
  \fancyhead[L,C]{}
  \fancyhead[R]{}{\raisebox{15pt}[0pt][0pt]{\xbox{\usebox{\LogoBox}}}}
  }

\hypersetup{%
  pdfcopyright={CC-by-nc-nd},
  pdfcreator={LaTeX with acmart, hyperref, and EDBT modifications},
  pdfvolumenum={XX},
  pdfissuenum={X},
  pdfisbn={\EDBTISBN},
  pdfissn={\EDBTISSN}
}

}

\usepackage{enumitem}
\usepackage{tikz}
\usepackage{pgfplots}
\pgfplotsset{compat=1.18}
\usetikzlibrary{arrows.meta}
\usepackage{subcaption}
\usepackage{multirow}
\usepackage{xcolor} 
\usepackage{xspace}
\newcommand{\mypar}[1]{\vspace{2mm}\noindent\emph{#1}}
\newcommand{\Adult}{\texttt{Adult}\xspace}

\newcommand{\XORDS}{\texttt{XOR}\xspace}
\newcommand{\PrivPGM}{\texttt{Private-PGM}\xspace}

\newcommand{\PE}{\texttt{PE}\xspace}
\newcommand{\TabPE}{\texttt{Tab-PE}\xspace}
\newcommand{\BFEvo}{\texttt{Hard}\xspace}
\newcommand{\GraphSmooth}{\texttt{BF-Soft}\xspace}

\newcommand{\Hard}{\texttt{Hard}\xspace}

\newcommand{\HighOrderErr}{\texttt{HOErr}\xspace}
\newcommand{\PairwiseErr}{\texttt{PWErr}\xspace}

\newcommand{\Bank}{\texttt{Bank}\xspace}

\usepackage{algorithm}
\usepackage{algpseudocode}
\algnewcommand{\Input}[1]{\Statex \textbf{Input:} #1}
\algnewcommand{\Output}[1]{\Statex \textbf{Output:} #1}

\usepackage{dsfont}

\newcommand{\Wasp}{\texttt{WASP}\xspace}
\newcommand{\Secpe}{\texttt{SecPE}\xspace}
\newcommand{\Balldp}{\texttt{Ball-DP}\xspace}

\usepackage{placeins}
\usepackage{etoolbox}

\newif\ifextended 
\extendedtrue 
\newcommand{\appletter}[1]{%
  \ifstrequal{#1}{app:related}{A}{%
  \ifstrequal{#1}{app:pe-threshold}{B}{%
  \ifstrequal{#1}{app:algorithm}{C}{%
  \ifstrequal{#1}{app:categorical}{D}{%
  \ifstrequal{#1}{app:proofs}{E}{%
  \ifstrequal{#1}{app:experiments}{F}{%
  \ifstrequal{#1}{app:exp-policy}{F.1}{%
  \ifstrequal{#1}{app:exp-protocol}{F.2}{%
  \ifstrequal{#1}{app:exp-crossings}{F.3}{%
  \ifstrequal{#1}{app:exp-equal-noise}{F.4}{%
  \ifstrequal{#1}{app:exp-interior}{F.5}{%
  \ifstrequal{#1}{app:exp-xor-interaction}{F.6}{%
  \ifstrequal{#1}{app:exp-forecast}{F.7}{%
  \ifstrequal{#1}{app:exp-broader}{F.8}{%
  \ifstrequal{#1}{app:exp-trainsize}{F.9}{??}}}}}}}}}}}}}}}}
\newcommand{\appref}[1]{%
  \ifextended Appendix~\ref{#1}%
  \else Appendix~\appletter{#1} of~\cite{extendedversion}\fi}

\theoremstyle{acmdefinition}
\newtheorem{remark}{Remark}
\theoremstyle{acmplain}

\begin{document}

\title{Soft Voting for Policy-Aware Private Data Synthesis}


\author{Yingge Hu}
\orcid{XXXX-XXXX-XXXX-XXXX}
\email{yhu893@uwo.ca}
\affiliation{%
  \institution{Department of Computer Science, Western University}
  \city{London}
  \state{Ontario}
  \country{Canada}
}

\author{Gautham Ramesh Babu}
\orcid{XXXX-XXXX-XXXX-XXXX}
\email{grg2@uwo.ca}
\affiliation{%
  \institution{Department of Computer Science, Western University}
  \city{London}
  \state{Ontario}
  \country{Canada}
}

\author{Mostafa Milani}
\orcid{XXXX-XXXX-XXXX-XXXX}
\email{mostafa.milani@uwo.ca}
\affiliation{%
  \institution{Department of Computer Science, Western University}
  \city{London}
  \state{Ontario}
  \country{Canada}
}

\renewcommand{\shortauthors}{Hu et al.}

\begin{abstract}
Blowfish privacy relaxes differential privacy (DP) by protecting only
the attribute-value substitutions a data owner specifies as edges of a
policy graph. A sparser policy can reduce the noise required by a
mechanism, but only when the released statistic changes less across
protected substitutions than across arbitrary DP neighbors. We study
this question for evolutionary, nearest-neighbor DP synthesizers such
as Private Evolution (\PE) and its tabular instantiation \TabPE, which
score private records against a candidate population and release a
noisy vote histogram. Their \emph{hard} vote is constant inside each
candidate's decision region and jumps at its boundary. Its policy-specific sensitivity therefore equals the full worst-case
value whenever at least one protected substitution crosses a boundary,
regardless of how short that substitution is. Because every round we
examined contained such a substitution, the policy graph gave no
reduction in noise.

We propose \GraphSmooth, a temperature-smoothed \emph{soft} vote whose
response changes gradually with distance. Its sensitivity has a tight
closed-form bound in the policy graph's reach and the temperature,
independent of the number of candidates, and the bound can be computed
once before synthesis. It also predicts from the policy alone when
policy-aware smoothing cannot substantially reduce noise: protecting a
flat categorical or binary attribute drives the reach to its maximum.
On real and synthetic datasets under narrow numeric policies, BF-Soft reduces error relative to hard voting at strong privacy budgets, while the advantage reverses at weaker budgets. A public-data pilot predicts when soft voting is beneficial without spending private budget.
\end{abstract}


\maketitle

\section{Introduction}

Differentially private (DP) synthetic data lets a data owner share a
sensitive table while limiting how much the released output can change
when one individual's record changes~\cite{dwork2006calibrating}.
Under bounded substitution adjacency, every allowed replacement of one
record is protected by the same worst-case privacy guarantee. On large
or structured domains such as age, income, or ZIP codes, this
worst-case symmetry can contribute substantially to the utility gap
between private and non-private synthetic data.

Policy-based definitions relax this symmetry by allowing the data
owner to specify which distinctions require protection.
Pufferfish~\cite{kifer2014pufferfish} makes the protected secrets
explicit, and Blowfish privacy~\cite{he2014blowfish} instantiates this
idea for tabular data using a \emph{policy graph}: an edge between two
values marks a substitution that the mechanism must protect. A policy
might, for example, protect a change from age 34 to 35 while not
requiring the same protection between ages 34 and 70. A policy graph
therefore restricts the neighbor relation used to compute sensitivity,
so a release's policy-specific sensitivity never exceeds its standard
global sensitivity~\cite{he2014blowfish}. The inequality can,
however, be tight. A sparser graph improves utility only when the
released statistic changes less across protected edges than across
arbitrary DP neighbors. This occurs for structured linear workloads
such as range and marginal queries~\cite{haney2015design}, but not for
a function whose response remains at its worst case even on short
protected substitutions.

We study when this distinction matters for evolutionary,
nearest-neighbor DP synthesis. Rather than privately measuring a fixed
set of statistics, this paradigm reshapes a population of candidate
synthetic records round by round. Each real record \emph{votes} for a
candidate; the votes are aggregated into a histogram, noise is added,
and the most-supported candidates survive while the rest are replaced.
This voting step, introduced in Private
Evolution~\cite{lin2024privateevolution} (\PE) and extended to tabular
data by \TabPE~\cite{tran2026tabpe}, is where privacy noise enters.
Recent extensions modify the vote or the protection target along other
axes (Section~\ref{sec:related}), but they do not ask whether the
vote's response can reflect how close a protected substitution is
under a policy graph.

This question is important because a sparse policy graph is useful
only if the mechanism can exploit it. We therefore ask:
\emph{when does a Blowfish policy reduce the sensitivity of the
nearest-neighbor vote, and can the vote be redesigned so that its
sensitivity changes gradually with the policy?}

The answer depends on the vote itself, not on the mutation rule, the
number of rounds, or the domain. The standard \emph{hard} vote assigns
a record to its single nearest candidate. A protected substitution
therefore has only two effects: either the nearest candidate remains
the same, and the vote does not change, or the assignment crosses a
decision boundary, and the vote moves entirely to another candidate.
In the latter case the $\ell_2$ change is the full worst-case value
$\sqrt2$, regardless of how short the protected substitution is.
Thus the hard vote cannot exploit the \emph{degree} to which a policy
restricts substitutions: its policy-specific sensitivity is either
$0$ or $\sqrt2$. In every round we checked in our experiments, some
protected edge crossed a decision boundary, so the policy graph
reduced no noise.


We replace the hard vote with \GraphSmooth, a temperature-smoothed \emph{soft} vote that splits each record's vote across candidates according to a softmax of negative distance. A short substitution then shifts only a small part of the vote, while larger substitutions can shift more. We show that its policy-specific sensitivity is at most $\sqrt2\tanh(r/2\tau)$, where $r$ is the policy graph's reach, the largest distance between two protected values, and $\tau$ is the temperature. The bound grows with the reach but never exceeds the hard-vote value $\sqrt2$, so a policy protecting only short substitutions yields less noise. It is tight, independent of the number of candidates, and computed once before synthesis. Our proof builds on a total-variation inequality for exponential tilts~\cite{cohen2023hyperbolic}: softmax ignores a common shift in its scores, so only their relative change matters, and the reach bounds it for all candidates at once.

The result also clarifies when policy-aware smoothing cannot help. Reach is a maximum over protected attributes rather than a sum, so
adding another narrowly protected numeric attribute need not increase
the bound. In contrast, protecting a flat categorical or binary attribute contributes that attribute's maximum distance and can make
the bound essentially equal to the hard vote's worst case at useful
temperatures. Participation (add/remove) edges similarly impose a floor that substitution policies cannot remove. Temperature creates a
second trade-off: increasing $\tau$ reduces sensitivity but also flattens the vote and weakens the evolutionary selection signal. The
bounds therefore do not determine an optimal temperature, and we use
a public-data pilot to select it without spending private budget.

Empirically, \GraphSmooth reduces error relative to hard voting at
strong privacy budgets under narrow numeric policies. On the UCI
\Adult and \Bank benchmarks, the crossover occurs at approximately $\epsilon \approx 0.16$--$0.46$ at $\delta = 10^{-5}$ for the fixed temperatures
we study; beyond this region, the cost of smoothing can dominate the
saved noise. Our experiments also isolate these two effects, test
policies protecting multiple attributes, and compare the analytical
sensitivity bound with exact sensitivity.

Our contributions are as follows.

\begin{itemize}[leftmargin=1.2em, itemsep=2pt, topsep=2pt]

\item \emph{A diagnosis of hard voting under policies.} For a linear
query under unconstrained bounded Blowfish adjacency, policy-specific
sensitivity reduces to a maximum over the policy graph's edges. For
the hard vote this maximum is all-or-nothing: it jumps from $0$ to
$\sqrt2$ at the first protected boundary crossing
(Section~\ref{sec:theory}).

\item \emph{\GraphSmooth,} a soft vote whose policy-specific
sensitivity has a tight closed-form bound in the policy's reach and
the temperature. The bound is independent of the candidate count,
can be computed once before synthesis, and identifies policies for
which smoothing cannot substantially reduce noise
(Sections~\ref{sec:method} and~\ref{sec:theory}).

\item \emph{An analysis of the temperature trade-off.} We bound both
its effect on sensitivity and its effect on the vote's ability to
discriminate among candidates, and use a public-data pilot to select
the temperature without spending private budget.

\item \emph{An evaluation} on the UCI \Adult and \Bank benchmarks and
\XORDS, a synthetic generator with planted higher-order (XOR)
dependencies, that locates where policy-aware soft voting helps and
where the advantage reverses, across privacy budgets and single- and
multi-attribute policies. The experiments separate the benefit of
lower noise from the cost of changing the vote and evaluate the
tightness and cost of the calibration bound
(Section~\ref{sec:experiments}).

\end{itemize}

Sections~\ref{sec:prelim} and~\ref{sec:related} review background and
related work, respectively, Section~\ref{sec:method} presents the
construction, Section~\ref{sec:theory} its analysis, and
Section~\ref{sec:experiments} the evaluation. The extended version
with all appendices is available in~\cite{extendedversion}, and our
code in~\cite{blowfish-evo-synth-repo}.
\section{Preliminaries}
\label{sec:prelim}

We review the privacy definitions and synthesis paradigm we use.

\subsection{Differential Privacy and zCDP}

A database $D$ is a multiset of $n$ records from a discrete domain
$\mathcal{T}$, with histogram $h(D) \in \mathbb{N}^{|\mathcal{T}|}$.
Databases $D, D'$ are \emph{neighbors}, written $D \sim D'$, if they
differ in one record. Under \emph{bounded} adjacency, $n$ is public
and one record's value changes. Under \emph{unbounded} adjacency, one
record is added or removed.

\begin{definition}[Differential Privacy~\cite{dwork2006calibrating}]
\label{def:dp}
A randomized mechanism $\mathcal{M}$ satisfies $\epsilon$-DP if for
all neighboring $D \sim D'$ and all measurable $S \subseteq
\mathrm{Range}(\mathcal{M})$,
$\Pr[\mathcal{M}(D) \in S] \le e^{\epsilon}\Pr[\mathcal{M}(D') \in S]$.
\end{definition}

For $f : \mathcal{T}^n \to \mathbb{R}^d$, the \emph{global
sensitivity} $\Delta_p f = \max_{D \sim D'} \|f(D) - f(D')\|_p$
calibrates the noise added to $f(D)$. Laplace noise scaled to
$\Delta_1 f/\epsilon$ gives $\epsilon$-DP, and Gaussian noise with
$\sigma^2 = (\Delta_2 f)^2/(2\rho)$ gives $\rho$-zCDP. We account for
privacy with zero-concentrated DP~\cite{bun2016concentrated}.

\begin{definition}[zCDP~\cite{bun2016concentrated}]
\label{def:zcdp}
A mechanism $\mathcal{M}$ satisfies $\rho$-zCDP if for all neighboring
$D \sim D'$ and all $\alpha \in (1,\infty)$, the R\'{e}nyi divergence
of order $\alpha$ between $\mathcal{M}(D)$ and $\mathcal{M}(D')$ is
at most $\rho\alpha$.
\end{definition}

zCDP is closed under post-processing and composes additively: an
adaptively chosen sequence of mechanisms with budgets
$\rho_1,\dots,\rho_k$ satisfies $\sum_i \rho_i$-zCDP. A $\rho$-zCDP
mechanism satisfies $(\epsilon,\delta)$-DP for every $\delta > 0$
with~\cite{bun2016concentrated}
\begin{equation}
\label{eq:zcdp-conversion}
\epsilon = \rho + 2\sqrt{\rho \ln(1/\delta)}.
\end{equation}
The synthesizer we study computes each round's candidates from
earlier noisy outputs. We therefore analyze each round conditional on
its candidate population and compose adaptively
(Section~\ref{sec:theory-setup}).

\subsection{Pufferfish and Blowfish Privacy}

Pufferfish~\cite{kifer2014pufferfish} replaces DP's fixed neighbor
relation with an explicit set of protected secret pairs.
Blowfish~\cite{he2014blowfish} instantiates this idea for tabular data
with a policy graph.

\begin{definition}[Policy graph~\cite{he2014blowfish}]
\label{def:policy-graph}
A policy is a triple $P = (\mathcal{T}, G, \mathcal{I}_Q)$, where
$G = (V,E)$ with $V \subseteq \mathcal{T} \cup \{\bot\}$ is the
\emph{policy graph}. An edge $(u,v) \in E$ marks values $u,v$ that
must remain indistinguishable. An edge $(u,\bot)$ to a dummy vertex
marks the presence of a record with value $u$ as protected.
$\mathcal{I}_Q$ is the set of databases consistent with public
deterministic constraints $Q$.
\end{definition}

Throughout, $\mathcal{I}_Q$ is unconstrained, and every result in
Section~\ref{sec:theory} assumes this.

\begin{definition}[Policy neighbors and policy-specific sensitivity~\cite{he2014blowfish}]
\label{def:policy-sens}
$D, D' \in \mathcal{I}_Q$ are \emph{policy neighbors},
$D \sim_P D'$, if they are minimally different and their difference
corresponds to an edge of $G$. An edge $(u,v)$ replaces one record's
value $u$ by $v$, and an edge $(u,\bot)$ adds or removes a record
with value $u$. Let $N(P)$ be the set of all such pairs. For
$f : \mathcal{I}_Q \to \mathbb{R}^d$ and $p \ge 1$, the
policy-specific sensitivity is
\[
S_p(f,P) = \max_{(D,D') \in N(P)} \|f(D) - f(D')\|_p.
\]
\end{definition}

\begin{definition}[Blowfish privacy~\cite{he2014blowfish} and its zCDP form]
\label{def:blowfish}
A mechanism satisfies $(\epsilon,P)$-Blowfish privacy if the
condition of Definition~\ref{def:dp} holds for all $D \sim_P D'$. It
satisfies $(\rho,P)$-zCDP Blowfish privacy if the condition of
Definition~\ref{def:zcdp} holds for all $D \sim_P D'$.
\end{definition}

The proofs of post-processing, composition, and
Eq.~\eqref{eq:zcdp-conversion} use only the divergence bound for a
fixed neighboring pair, never the structure of the neighbor relation.
All three therefore carry over unchanged to $(\rho,P)$-zCDP Blowfish
privacy. In particular, $T$ rounds with budgets $\rho_t$ give
$(\sum_t \rho_t, P)$-zCDP Blowfish privacy, which implies
$(\epsilon,\delta,P)$-Blowfish privacy with $\epsilon$ as in
Eq.~\eqref{eq:zcdp-conversion}.

\begin{theorem}[\cite{he2014blowfish,bun2016concentrated,dwork2006calibrating}]
\label{thm:blowfish-laplace}
Releasing $f(D) + \eta$ with
$\eta_i \sim \mathrm{Lap}(S_1(f,P)/\epsilon)$ i.i.d.\ satisfies
$(\epsilon,P)$-Blowfish privacy. Releasing
$f(D) + \mathcal{N}(0,\sigma^2 I)$ with
$\sigma^2 = S_2(f,P)^2/(2\rho)$ satisfies $(\rho,P)$-zCDP Blowfish
privacy. When $G$ is the complete graph on $\mathcal{T}$ and
$\mathcal{I}_Q$ is unconstrained, both coincide with their
standard bounded-DP counterparts.
\end{theorem}

Since $N(P)$ contains only standard neighbor pairs,
$S_p(f,P) \le \Delta_p f$ for every $G$. For the
sensitivity-calibrated mechanisms we use, a Blowfish policy reduces
noise only when this inequality is strict. A sparse policy graph is
therefore not sufficient by itself: the released function must also
respond less to the substitutions represented by its edges than to
arbitrary DP neighbors. Section~\ref{sec:theory} characterizes this
condition for the voting functions we study.

Values not joined by an edge are still protected along paths of $G$.
A shortest path of length $\ell$ gives $\ell\epsilon$ under
$(\epsilon,P)$-Blowfish privacy and, by group privacy for
zCDP~\cite{bun2016concentrated}, $\ell^2\rho$ under
$(\rho,P)$-zCDP Blowfish privacy. Values in different connected
components are not protected from each other~\cite{he2014blowfish}.

\paragraph{Specifying policies.}
A data owner typically specifies one graph per attribute of a product
domain $\mathcal{T} = \mathcal{T}_1 \times \cdots \times
\mathcal{T}_p$. These combine into a record-level graph $G$ that
links two records agreeing on every attribute but one, on which they
differ by a protected pair. Our experiments
(Section~\ref{sec:experiments}) use three common forms:

\begin{itemize}
\item \emph{Threshold graphs} protect nearby values under a numeric,
ordinal, or hierarchical distance: values within distance $\theta$
are adjacent, while more distant values are not directly protected
(Definition~\ref{def:threshold-graph}; the hierarchical case is in
\appref{app:categorical}).

\item \emph{Block-partition graphs} protect substitutions within a
coarse category but not across categories. They can, for example,
hide which value inside a category an individual has while allowing
the category itself to remain distinguishable.

\item \emph{Complete graphs} connect every pair of values of an
attribute and therefore protect that attribute without distinguishing
among its values.
\end{itemize}

Edges to $\bot$ express unbounded adjacency in the same formalism;
which values receive them is a modeling choice
(Section~\ref{sec:bounded-unbounded}). Our experiments use bounded
adjacency. Results in Section~\ref{sec:theory} hold for any $G$
unless stated otherwise.

\begin{definition}[Threshold policy graph and graph reach]
\label{def:threshold-graph}
Let $d$ be a distance on $\mathcal{T}$. For $\theta \ge 0$, the
\emph{threshold policy graph} $G_\theta = (\mathcal{T}, E_\theta)$
has edge set
$E_\theta = \{(u,v) : 0 < d(u,v) \le \theta\}$. For any policy graph
$G = (\mathcal{T}, E)$, the \emph{graph reach} with respect to $d$ is
$r_d(G) = \max_{(u,v) \in E} d(u,v)$, so $r_d(G_\theta) = \theta$
whenever some pair attains the threshold.
\end{definition}

Graph reach summarizes the longest substitution that the policy
protects under the same distance used by the vote. It does not
describe the entire topology of $G$, but Section~\ref{sec:theory-softmax}
shows that it is sufficient for a population-independent sensitivity
bound for our soft vote.

\subsection{Evolutionary, Nearest-Neighbor Synthesis}
\label{sec:prelim-evo}

Private Evolution~\cite{lin2024privateevolution} (\PE) and its tabular
instantiation \TabPE~\cite{tran2026tabpe} evolve a population of
candidate records rather than measuring marginals of the data. The
initial population is generated without private data. At round $t$,
with population $C_t = \{c_1,\dots,c_m\}$, each private record
$x \in D$ votes for its nearest candidate under a base distance $d$ (e.g., Gower distance~\cite{gower1971general} for mixed attributes):
\[
v_{C_t}(x) = e_{\mathrm{nn}(x)} \in \{0,1\}^m,
\qquad
\mathrm{nn}(x) = \arg\min_j d(x,c_j),
\]
where $e_j$ is the $j$th standard basis vector. This is the
\texttt{DP\_NN\_HISTOGRAM} step of \PE and \TabPE in vector form. The
released statistic is the histogram
$H_t(D;C_t) = \sum_{x \in D} v_{C_t}(x)$, noised as
$\tilde H_t = H_t(D;C_t) + \mathcal{N}(0,\sigma_t^2 I)$ under a
per-round budget $\rho_t$. A resampling and mutation step,
\texttt{ResampleMutate}, reads only $\tilde H_t$ and $C_t$ to form
$C_{t+1}$, and the loop repeats for $T$ rounds.

Our analysis (Sections~\ref{sec:method} and~\ref{sec:theory}) depends
only on the form of $v_{C_t}$, not on $d$, the mutation rule, or $T$.
We evaluate it in a \TabPE-style instantiation
(Section~\ref{sec:experiments}). Because $v_{C_t}$ is one-hot, its
response to a substitution takes only two values.
Section~\ref{sec:theory-hardvote} shows that this prevents sensitivity
from varying gradually with the policy's notion of closeness.
\section{Related Work}
\label{sec:related}

We position our work against DP synthesizers, policy-aware mechanisms,
metric relaxations of DP, and smoothed assignment rules.

\mypar{Evolutionary synthesis}
\PE~\cite{lin2024privateevolution} evolves a candidate population
through private nearest-neighbor voting against the real data,
without training a model. Its tabular instantiation
\TabPE~\cite{tran2026tabpe} reports gains over marginal-based
baselines on data with high-order correlations that low-order
marginals miss. Both use a one-hot vote each round. Under a policy
graph, this vote remains constant while the nearest candidate is
unchanged and jumps when a substitution crosses a decision boundary,
which makes its policy-specific sensitivity a step function
(Section~\ref{sec:theory-hardvote}).

Two extensions relax \PE along axes orthogonal to ours.
\Wasp~\cite{zou2025wasp} spreads a record's vote over its top-$Q$
candidates with weights determined by rank. This makes the vote less
concentrated, but rank is itself piecewise constant in the underlying
distances: a protected substitution can leave every rank unchanged or
change a rank discontinuously at a boundary. Its response is therefore
still step-like rather than graded by the substitution distance. In
contrast, our soft vote changes continuously with distance.
\Wasp changes how many candidates share a vote; we change how the
vote responds to a substitution, and the two could in principle
compose. \appref{app:related} gives the fuller comparison.

\Secpe~\cite{wang2026secpe} relaxes \PE's protection target rather
than its vote, replacing membership privacy with a secret-specific
$(p,r)$-protection guarantee; it has no policy graph.

\mypar{Other DP synthesizers}
Marginal-based mechanisms such as
PrivBayes~\cite{zhang2017privbayes},
Private-PGM~\cite{mckenna2019graphical},
MST~\cite{mckenna2021winning}, and AIM~\cite{mckenna2022aim} privately
measure selected low-order marginals and fit a joint distribution to
them. GAN-based mechanisms~\cite{xie2018dpgan,jordon2019pategan,fang2022dpctgan}
add noise during training rather than to a fixed statistic. Both
families target standard DP. Our goal is to isolate how a released
statistic can exploit a policy within one synthesis paradigm. A
utility comparison across paradigms under a common policy would
require policy-aware versions of the competing synthesizers, which we
leave open.

\mypar{Policy-aware mechanisms}
Haney et al.~\cite{haney2015design} answer a fixed linear workload
under Blowfish by transforming the workload and policy graph so that
standard DP mechanisms apply. Conditional on $C_t$, one round's vote
histogram is itself a linear workload, and our sensitivity identity
(Lemma~\ref{lem:policy-sensitivity}) is consistent with this view.
The difference is where the design freedom lies. Their workload is
given and the mechanism is optimized for answering it. In
evolutionary synthesis, the per-record voting map determines the
workload and the candidate population changes adaptively across
rounds. We redesign this map itself so that its response to a
protected substitution reflects the policy's distance, rather than
only recalibrating noise for a fixed hard-vote workload.

Xiao et al.~\cite{xiao2016dphmm} study the complementary problem of
when temporal correlation undermines the protection intended by a
policy graph.

\mypar{Metric relaxations}
Blowfish is one instance of relaxing DP's neighbor relation according
to an application-specific notion of closeness, as in metric
DP~\cite{chatzikokolakis2013broadening}, Lipschitz
privacy~\cite{koufogiannis2015lipschitz}, and
geo-indistinguishability~\cite{andres2013geoindistinguishability}.
Geo-indistinguishability, for example, uses a mechanism whose output
distribution varies continuously with geographic distance.
\Balldp~\cite{margaryan2026balldp} restricts substitutions to a fixed
radius, which corresponds to a threshold policy graph
(Definition~\ref{def:threshold-graph}), but applies the restriction to
convex empirical risk minimization rather than to voting or synthesis.
The fuller distinction is given in \appref{app:related}.

\mypar{Smoothed assignment rules}
Temperature-scaled softmax is a standard smooth relaxation of a hard
$\arg\max$ or $\arg\min$ rule, and related sensitivity or Lipschitz
bounds appear in the exponential
mechanism~\cite{mcsherry2007mechanism,optimal-approx-softmax} and in
bounds on the softmax map itself~\cite{nair2026softmax}. Applied
directly to our logits, a generic $\ell_2$ argument controls the
output change by the $\ell_2$ size of the logit perturbation. Although
each coordinate changes by at most $d(x,y)/\tau$, this norm can be as
large as $\sqrt{m}\,d(x,y)/\tau$ for $m$ candidates.

Our analysis exploits additional structure of a policy substitution.
Softmax is invariant to a common shift of all logits, so only their
\emph{oscillation} matters. A substitution from $x$ to $y$ has
oscillation at most $2d(x,y)/\tau$, independent of $m$. Combining
this observation with the total-variation inequality of Cohen and
Fausti~\cite{cohen2023hyperbolic} yields the tight, saturating bound
of Corollary~\ref{cor:softmax}; Proposition~\ref{prop:lower-bound}
shows that no uniformly smaller bound depending only on policy reach
and temperature is possible. \appref{app:related} gives the detailed
comparison with generic softmax bounds.
\section{Policy-Aware Evolutionary Synthesis}
\label{sec:method}

We replace the hard nearest-neighbor vote with a policy-aware
alternative while leaving the rest of the evolutionary synthesis
paradigm unchanged.

\subsection{Generalizing the Vote}
\label{sec:method-generalize}

The hard vote $v_{C_t}(x) = e_{\mathrm{nn}(x)}$ puts all of a
record's mass on its nearest candidate, a vertex of the probability
simplex $\Delta^{m-1}$ over the $m$ candidates in $C_t$. Neither noise
calibration nor resampling requires the vote to be one-hot: both only
need the per-record responses to aggregate into
\[
H_t(D;C_t) = \sum_{x\in D} v_{C_t}(x).
\]
We therefore allow any simplex-valued response,
\[
v_{C_t}(x) \in \Delta^{m-1}
=
\left\{
p \in \mathbb{R}^m_{\ge0}
:
\textstyle\sum_j p_j = 1
\right\},
\]
of which the hard vote is the special case confined to the simplex
vertices.

This abstraction separates the \emph{per-record response} from the
rest of the evolutionary loop. Once the responses are aggregated and
privately released, candidate resampling and mutation are
post-processing. We can therefore change the vote and its sensitivity
calibration without changing the remainder of the synthesis
procedure.

\subsection{Blowfish-Compatible Voting}
\label{sec:method-bfevo}
\label{sec:method-adjacency}

We define \BFEvo as the evolutionary paradigm run under a policy $P$.
It replaces standard adjacency with $\sim_P$ and calibrates each
round's noise to the policy-specific sensitivity of that round's
histogram, for any simplex-valued vote. At round $t$, conditional on
$C_t$, it releases
\[
\tilde H_t
=
H_t(D;C_t) + \mathcal{N}(0,\sigma_t^2 I),
\qquad
\sigma_t^2
=
\frac{S_2(H_t,P)^2}{2\rho_t},
\]
where $S_2(H_t,P)$ is taken over the edges of $G$ with $C_t$ fixed.
By Theorem~\ref{thm:blowfish-laplace} and adaptive composition
(Section~\ref{sec:prelim}), $T$ rounds satisfy
$(\sum_t \rho_t, P)$-zCDP Blowfish privacy, and hence
$(\epsilon,\delta,P)$-Blowfish privacy via
Eq.~\eqref{eq:zcdp-conversion}.

This gives two distinct ways to make the evolutionary mechanism
policy-aware. We can keep the original hard vote and only replace its
global sensitivity by its policy-specific sensitivity; we call this
method \emph{Hard}. Or we can also redesign the per-record response
so that protected substitutions induce smaller changes in the vote;
this gives \GraphSmooth below. \TabPE, Hard, and \GraphSmooth
therefore share the same evolutionary loop but differ in their vote
and/or its calibration.

We use bounded adjacency throughout. Participation edges, which
encode add/remove adjacency in the same graph, are treated in
Section~\ref{sec:bounded-unbounded}. Practical \PE implementations
also zero out small entries of $\tilde H_t$ with a public threshold.
This is post-processing and does not affect sensitivity
(\appref{app:pe-threshold}).

\subsection{Hard and Soft Votes}
\label{sec:method-hardvote}
\label{sec:method-graphsmooth}

\mypar{Hard vote}
A protected substitution either keeps a record's vote on the same
candidate, at no cost, or moves it entirely to another candidate. In
the second case it incurs the full $\ell_2$ distance $\sqrt2$ between
two one-hot vectors, regardless of how short the substitution is.

Figure~\ref{fig:hardvote-voronoi} shows this on $C_t$'s Voronoi
cells: solid policy edges whose endpoints lie in the same cell cost
$0$, while the dotted edge crossing a cell boundary costs $\sqrt2$.
Thus geometric closeness of a protected pair does not translate into
closeness of the released response once the pair crosses a decision
boundary. Corollary~\ref{cor:indicator} makes this exact: the hard
vote's sensitivity is $\sqrt2$ whenever any protected edge crosses a
boundary, and $0$ otherwise.

Whether a crossing exists can change from round to round as the
candidate population evolves (Section~\ref{sec:theory-hardvote});
Table~\ref{tab:hard-vote} gives the analogue for other adjacency notions.
In every round we checked in our experiments, some protected edge
crossed a boundary, so Hard's calibration equaled the raw $\sqrt2$ of
\TabPE (Section~\ref{sec:experiments}).

\begin{figure}[t]
\centering
\begin{subfigure}[b]{0.48\linewidth}
\centering
\resizebox{\linewidth}{!}{%
\begin{tikzpicture}[scale=1.05]
  \definecolor{c1fill}{RGB}{255,235,205}
  \definecolor{c2fill}{RGB}{205,232,255}
  \definecolor{c3fill}{RGB}{223,255,214}
  \definecolor{c4fill}{RGB}{255,214,235}
  \fill[c1fill] (-0.592,-0.422) -- (3.329,-0.422) -- (2.940,1.807) -- (-0.592,2.899) -- cycle;
  \fill[c2fill] (3.329,-0.422) -- (6.630,-0.422) -- (6.630,2.221) -- (4.206,2.467) -- (2.940,1.807) -- cycle;
  \fill[c3fill] (-0.592,4.392) -- (-0.592,2.899) -- (2.940,1.807) -- (4.206,2.467) -- (3.971,4.392) -- cycle;
  \fill[c4fill] (3.971,4.392) -- (4.206,2.467) -- (6.630,2.221) -- (6.630,4.392) -- cycle;
  \draw[gray!55, thick] (3.329,-0.422) -- (2.940,1.807);
  \draw[gray!55, thick] (-0.592,2.899) -- (2.940,1.807);
  \draw[gray!55, thick] (2.940,1.807) -- (4.206,2.467);
  \draw[gray!55, thick] (4.206,2.467) -- (6.630,2.221);
  \draw[gray!55, thick] (4.206,2.467) -- (3.971,4.392);
  \fill[black] (1.184,0.845) circle (2.6pt);
  \node[fill=c1fill, inner sep=1.5pt, font=\large] at (1.539,1.098) {$c_{1}$};
  \fill[black] (4.972,1.183) circle (2.6pt);
  \node[fill=c2fill, inner sep=1.5pt, font=\large] at (5.328,1.436) {$c_{2}$};
  \fill[black] (2.723,3.379) circle (2.6pt);
  \node[fill=c3fill, inner sep=1.5pt, font=\large] at (3.078,3.632) {$c_{3}$};
  \fill[black] (5.446,3.548) circle (2.6pt);
  \node[fill=c4fill, inner sep=1.5pt, font=\large] at (5.801,3.801) {$c_{4}$};
  \draw[black, very thick]                       (1.539,0.084) -- (2.368,0.338);
  \draw[red!70!black, very thick, densely dotted] (2.368,0.338) -- (3.907,1.140);
  \draw[black, very thick]                       (3.907,1.140) -- (4.617,1.647);
  \draw[red!70!black, very thick, densely dotted] (4.617,1.647) -- (5.150,2.872);
  \fill[orange!80!black] (1.539,0.084) circle (1.8pt);
  \fill[orange!80!black] (2.368,0.338) circle (1.8pt);
  \fill[blue!60!black]   (3.907,1.140) circle (1.8pt);
  \fill[blue!60!black]   (4.617,1.647) circle (1.8pt);
  \fill[purple!70!black] (5.150,2.872) circle (1.8pt);
  \node[font=\large, anchor=north] at (1.953,0.084) {sens.\ $0$};
  \node[red!70!black, font=\large, anchor=north] at (3.5,0.75) {sens.\ $\sqrt{2}$};
  \node[font=\large, anchor=south] at (3.966,1.647) {sens.\ $0$};
  \node[red!70!black, font=\large, anchor=west]  at (4.9,2.238) {sens.\ $\sqrt{2}$};
\end{tikzpicture}%
}
\caption{Hard vote: $0$ inside a cell, $\sqrt2$ across a boundary.}
\label{fig:hardvote-voronoi}
\end{subfigure}
\hfill
\begin{subfigure}[b]{0.48\linewidth}
\centering
\resizebox{\linewidth}{!}{%
\begin{tikzpicture}[scale=1.05]
  \definecolor{c1fill}{RGB}{255,235,205}
  \definecolor{c2fill}{RGB}{205,232,255}
  \definecolor{c3fill}{RGB}{223,255,214}
  \definecolor{c4fill}{RGB}{255,214,235}
  \node[anchor=south west, inner sep=0] at (-0.592,-0.422)
    {\includegraphics[width=7.58cm,height=5.1cm]{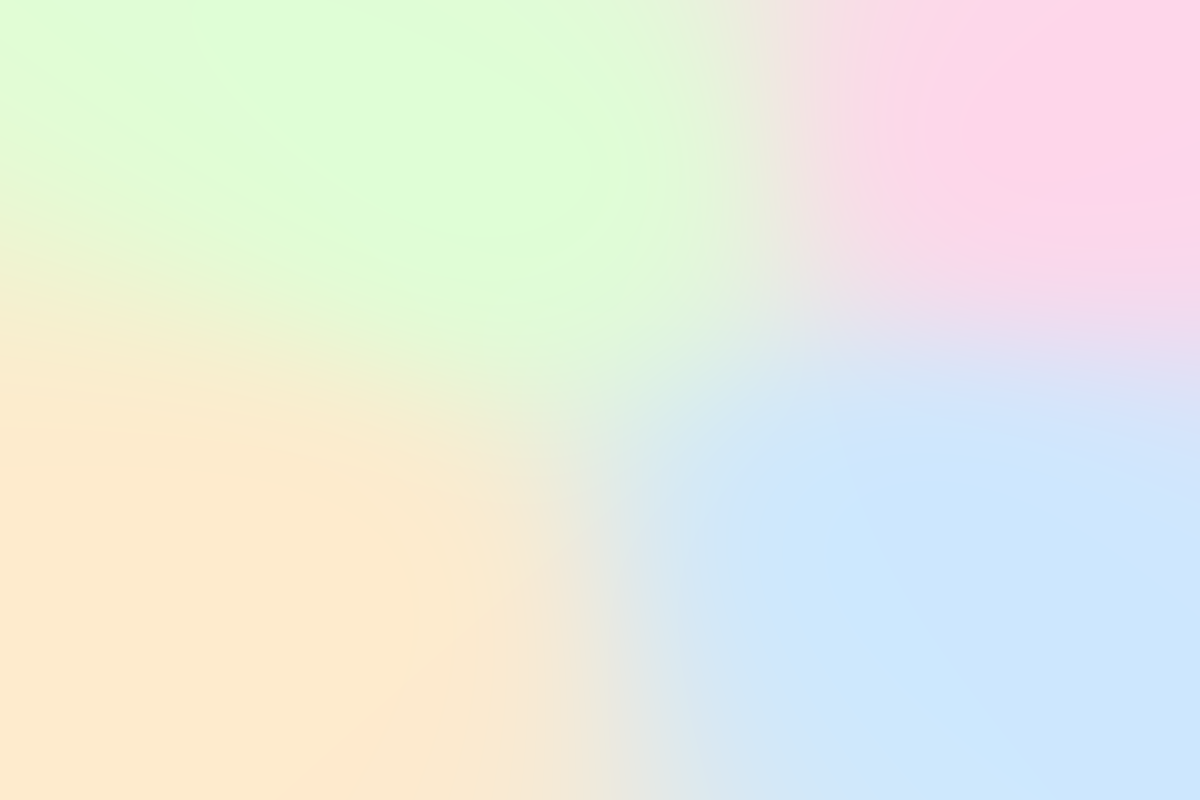}};
  \draw[gray!55, thick] (3.329,-0.422) -- (2.940,1.807);
  \draw[gray!55, thick] (-0.592,2.899) -- (2.940,1.807);
  \draw[gray!55, thick] (2.940,1.807) -- (4.206,2.467);
  \draw[gray!55, thick] (4.206,2.467) -- (6.630,2.221);
  \draw[gray!55, thick] (4.206,2.467) -- (3.971,4.392);
  \fill[black] (1.184,0.845) circle (2.6pt);
  \node[fill=c1fill, inner sep=1.5pt, font=\large] at (1.539,1.098) {$c_{1}$};
  \fill[black] (4.972,1.183) circle (2.6pt);
  \node[fill=c2fill, inner sep=1.5pt, font=\large] at (5.328,1.436) {$c_{2}$};
  \fill[black] (2.723,3.379) circle (2.6pt);
  \node[fill=c3fill, inner sep=1.5pt, font=\large] at (3.078,3.632) {$c_{3}$};
  \fill[black] (5.446,3.548) circle (2.6pt);
  \node[fill=c4fill, inner sep=1.5pt, font=\large] at (5.801,3.801) {$c_{4}$};
  \draw[black, very thick, densely dotted] (1.539,0.084) -- (2.368,0.338);
  \draw[black, very thick, densely dotted] (2.368,0.338) -- (3.907,1.140);
  \draw[black, very thick, densely dotted] (3.907,1.140) -- (4.617,1.647);
  \draw[black, very thick, densely dotted] (4.617,1.647) -- (5.150,2.872);
  \fill[orange!80!black] (1.539,0.084) circle (1.8pt);
  \fill[orange!80!black] (2.368,0.338) circle (1.8pt);
  \fill[blue!60!black]   (3.907,1.140) circle (1.8pt);
  \fill[blue!60!black]   (4.617,1.647) circle (1.8pt);
  \fill[purple!70!black] (5.150,2.872) circle (1.8pt);
  \node[black, font=\large, anchor=north] at (1.953,0.084) {$0.09$};
  \node[black, font=\large, anchor=north] at (3.5,0.75) {$1.21$};
  \node[black, font=\large, anchor=south] at (3.966,1.647) {$0.06$};
  \node[black, font=\large, anchor=west]  at (4.9,2.238) {$1.19$};
\end{tikzpicture}%
}
\caption{Soft vote ($\tau=0.5$): grades smoothly to the $\sqrt2$ ceiling.}
\label{fig:softvote-voronoi}
\end{subfigure}
\vspace{-2mm}
\caption{Sensitivity along the same threshold policy graph and
  candidate population, hard vote (left) vs.\ soft vote (right).}
\label{fig:voronoi-comparison}
\end{figure}

\mypar{Soft vote}
\GraphSmooth replaces the one-hot assignment with a
temperature-smoothed vote under the same base distance $d$:
\begin{equation}
\label{eq:softvote}
v_{C_t}(x)[j]
=
\frac{\exp(-d(x,c_j)/\tau)}
     {\sum_k \exp(-d(x,c_k)/\tau)}.
\end{equation}
For $\tau>0$, the response varies continuously with the candidate
distances. A short substitution can therefore move only a small
amount of voting mass, while a larger substitution can move more,
up to the $\sqrt2$ diameter attained by two different one-hot votes.
Figure~\ref{fig:softvote-voronoi} illustrates this graded response on
the same policy graph used for the hard vote.

Reading the two panels together: edges that cost $0$ under the hard
vote cost only $0.06$--$0.09$ under the soft vote, while the
boundary-crossing edges cost about $1.2$ instead of the full
$\sqrt2\approx1.41$. The soft vote thus pays a little everywhere to
avoid paying the maximum anywhere.

Sensitivity is still the worst response over all protected pairs, but
the response to an individual edge is no longer restricted to
$\{0,\sqrt2\}$. Section~\ref{sec:theory-softmax} shows that, for every
candidate population,
\[
S_2(H_t,P)
\le
\sqrt2\tanh\!\left(\frac{r_d(G)}{2\tau}\right).
\]
The bound decreases smoothly as the policy reach $r_d(G)$ shrinks or
the temperature $\tau$ grows, and is independent of the number of
candidates. Unlike exact sensitivity, it also does not depend on the
realized population $C_t$.

\mypar{Categorical attributes}
\label{sec:method-categorical}
The bound measures policy edges using the same base distance $d$ that
defines the vote. For numeric and ordinal attributes, a policy can
therefore restrict the permitted substitution distance directly.
Under Gower's default categorical term, however, every pair of
distinct values has the same nonzero distance. Protecting a flat
categorical or binary attribute consequently contributes its full
Gower weight to the reach, the largest contribution that attribute
can make.

The analytical bound remains strictly below $\sqrt2$ for finite
reach and $\tau>0$ (Remark~\ref{rem:strict-gap}), but at temperatures
that retain a sharp selection signal it can be numerically very close
to $\sqrt2$, leaving little noise reduction. When a categorical
attribute has genuine hierarchical structure, replacing the flat
distance by a tree metric restores several levels of policy distance.
The remedy lies in $d$, not in the vote or the policy graph: the soft
vote and the bound of Corollary~\ref{cor:softmax} assume only that
$d$ is a metric. For values organized in a rooted taxonomy (job codes,
product categories, administrative regions), the tree-path distance
$d(u,v)=\mathrm{depth}(u)+\mathrm{depth}(v)-2\,\mathrm{depth}(\mathrm{lca}(u,v))$
takes the nonzero values $2,4,\dots,2L$ when all leaves are at depth
$L$. A threshold policy can then protect only substitutions within a
subtree. After normalizing by the diameter $2L$, protecting sibling
substitutions gives reach $w/L$ instead of $w$; for \Bank's
\texttt{job} ($L=2$, four groups), this halves the reach from $0.2$ to
$0.1$. A shallow hierarchy thus recovers gradation, though not the
small reach of a narrow numeric threshold such as \texttt{age}'s
$0.017$. A flat block partition is the depth-$1$ case, and a binary
attribute admits no gradation under any metric.

\mypar{Temperature}
The temperature $\tau$ affects more than sensitivity. Increasing it
flattens the vote toward the uniform distribution and weakens the
selection pressure available to \texttt{ResampleMutate}. It therefore
trades privacy-noise reduction against the usefulness of the voting
signal. Section~\ref{sec:theory-convergence} makes both limits
explicit, and Section~\ref{sec:experiments} measures the resulting
utility trade-off.

\subsection{Instantiation and Calibration Cost}
\label{sec:method-instantiation}
\label{sec:method-algorithm}

\BFEvo and \GraphSmooth change only the per-record vote and its
calibration, so they can be substituted into any implementation of
the evolutionary paradigm. In our \TabPE-style instantiation, $d$ is
a Gower distance and the mutation operator is tuned for tabular data.
The resampling rule, mutation schedule, candidate populations, and
privacy-budget schedule are otherwise unchanged.

Algorithm~\ref{alg:unified} makes the comparison explicit. The three
methods share the aggregation, Gaussian mechanism, and
\texttt{ResampleMutate} steps. They differ only at the vote
(line~\ref{alg:vote}) and the sensitivity calibration
(line~\ref{alg:sens}).

\begin{algorithm}[t]
\caption{Evolutionary DP Synthesis}
\label{alg:unified}
\begin{algorithmic}[1]
\Input{policy graph $G=(\mathcal{T},E)$, reach $r_d(G)$, distance $d$,
temperature $\tau$, budgets $\{\rho_t\}_{t=1}^T$,
$\texttt{method} \in \{\TabPE, \Hard, \GraphSmooth\}$}
\Output{final candidate population $C_{T+1}$}
\State $C_1 \gets \texttt{InitPopulation}();$
\For{$t = 1, \dots, T$}
  \For{$x \in D$} \label{alg:vote}
    \If{$\texttt{method} = \GraphSmooth$}
      \State $v_{C_t}(x) \gets
      \texttt{softmax}_j\big(-d(x,c_j)/\tau\big);$
      \Comment{Eq.~\eqref{eq:softvote}}
    \Else
      \State $v_{C_t}(x) \gets
      e_{\arg\min_j d(x,c_j)};$
      \Comment{$\mathrm{nn}(x)$}
    \EndIf
  \EndFor
  \State $H_t(D;C_t) \gets \sum_{x \in D} v_{C_t}(x);$

  \If{$\texttt{method} = \TabPE$} \label{alg:sens}
    \State $\Delta_t \gets \sqrt2;$
    \Comment{global sensitivity}
  \ElsIf{$\texttt{method} = \Hard$}
    \State $\Delta_t \gets
    \sqrt2 \cdot
    \mathds{1}\!\left[
    \exists\,(u,v)\in E :
    \mathrm{nn}_{C_t}(u)\ne\mathrm{nn}_{C_t}(v)
    \right];$
    \Comment{Cor.~\ref{cor:indicator}}
  \ElsIf{$\texttt{method} = \GraphSmooth$}
    \State $\Delta_t \gets
    \sqrt2\tanh\!\big(r_d(G)/2\tau\big);$
    \Comment{Cor.~\ref{cor:softmax}}
  \EndIf

  \State $\sigma_t^2 \gets \Delta_t^2/(2\rho_t);$
  \State $\tilde H_t \gets
  H_t(D;C_t)+\mathcal{N}(0,\sigma_t^2I);$
  \State $C_{t+1} \gets
  \texttt{ResampleMutate}(\tilde H_t,C_t);$
\EndFor
\State \Return $C_{T+1};$
\end{algorithmic}
\end{algorithm}

Table~\ref{tab:method-comparison} summarizes the resulting
differences. \TabPE and Hard use exactly the same one-hot response;
Hard differs only by restricting the sensitivity maximum to policy
edges. \GraphSmooth changes the response itself and calibrates to a
population-independent upper bound.

\begin{table}[t]
\centering
\small
\caption{The three instantiations of Algorithm~\ref{alg:unified}.}
\label{tab:method-comparison}
\resizebox{\columnwidth}{!}{%
\begin{tabular}{llll}
\toprule
Method & Per-record vote & Calibration & Recomputed each round? \\
\midrule
\TabPE
& one-hot
& $\sqrt2$
& no \\

\Hard/\BFEvo
& one-hot
& $0$ or $\sqrt2$
& yes, for exact calibration \\

\GraphSmooth
& softmax
& $\sqrt2\tanh(r_d(G)/2\tau)$
& no \\
\bottomrule
\end{tabular}
}
\end{table}

\mypar{Round-dependent versus closed-form calibration}
For Hard, exact calibration depends on the current population $C_t$.
One protected edge whose endpoints have different nearest candidates
is enough to certify sensitivity $\sqrt2$, but certifying sensitivity
$0$ requires ruling out such a crossing for every protected edge.
Since resampling and mutation change $C_t$, the same fixed policy can
give different exact sensitivities in different rounds.

For \GraphSmooth, the exact sensitivity also depends on $C_t$, and
can be substantially smaller than the analytical bound. The
difference is that the bound depends only on the policy reach
$r_d(G)$ and the fixed temperature $\tau$. It can therefore be
computed once before synthesis and remains valid for every candidate
population encountered by the adaptive process.

\mypar{Why not calibrate to exact soft sensitivity?}
Because $C_t$ is itself obtained from previous private releases,
conditioning on the realized $C_t$ and computing exact sensitivity
would not by itself violate privacy. The obstacle is computational.
By Lemma~\ref{lem:policy-sensitivity}, exact sensitivity requires
maximizing the change in the per-record response over every edge of
the \emph{record-level} policy graph. For a product domain, a
protected pair on one attribute appears once for every combination of
the remaining attributes. Thus the record-level edge set can grow
exponentially in the number of attributes even when each
per-attribute policy is simple.

The same search is needed to certify Hard's zero-sensitivity case.
By contrast, once $r_d(G)$ has been obtained from the per-attribute
policies, evaluating the \GraphSmooth bound is constant-time and
requires no inspection of $C_t$. Both hard nearest-neighbor voting
and softmax voting cost $O(m)$ per record, so the soft vote changes
no asymptotic voting cost. \appref{app:algorithm} gives the full
adjacency and calibration-cost discussion, and
Section~\ref{sec:exp-sensitivity-summary} evaluates the cost
empirically.
\section{Sensitivity of Hard and Soft Votes}
\label{sec:theory}

We now make Section~\ref{sec:method}'s claims precise. For any linear
query, policy-specific sensitivity reduces to a maximum over the
policy graph's edges (\S\ref{sec:theory-identity}). For the hard vote
this maximum is all-or-nothing (\S\ref{sec:theory-hardvote}), while
for the soft vote it admits a tight closed-form bound in the policy
reach and temperature (\S\ref{sec:theory-softmax}). Temperature
trades this reduction against the vote's ability to discriminate
among candidates (\S\ref{sec:theory-convergence}), and some policy
edges impose floors that reach-dependent calibration cannot lower
(\S\ref{sec:bounded-unbounded}). Full proofs are in
\appref{app:proofs}.

\subsection{Linear Queries}
\label{sec:theory-setup}

Let $(\mathcal{T},d)$ be a finite metric space and
$\phi : \mathcal{T} \to \mathbb{R}^m$ a per-record embedding. A
\emph{linear query}~\cite{haney2015design} is a release
\[
f(D)=\sum_{x\in D}\phi(x)=W\,h(D),
\]
linear in the database histogram $h(D)$ even when $\phi$ is nonlinear
in $x$.

At a fixed round, conditional on $C_t$ and $\tau$, both votes of
Section~\ref{sec:method} are linear queries. The hard vote has
$\phi(x)=e_{\mathrm{nn}(x)}$, while the soft vote has
$\phi(x)=v_{C_t}(x)$ from Eq.~\eqref{eq:softvote}. Raw marginal
counts are linear queries as well. Model-fitting procedures optimized
jointly over all records, such as \PrivPGM~\cite{mckenna2019graphical},
are not of this form.

All results below are therefore per-round statements conditional on
the realized $C_t$. This is sufficient for privacy accounting because
$C_t$ is computed from earlier private releases, and the rounds
compose adaptively. The policy $P$ itself is fixed before the
mechanism runs.

\subsection{A Sensitivity Identity}
\label{sec:theory-identity}

Replacing one record's value $u$ by $v$ changes a linear query by
exactly $\phi(u)-\phi(v)$; every other record cancels. Thus, under
bounded and unconstrained Blowfish adjacency, the worst case over
databases becomes a worst case over policy edges.

\begin{lemma}[Policy sensitivity is a maximum over edges]
\label{lem:policy-sensitivity}
For a linear query $f(D)=\sum_{x\in D}\phi(x)$, any policy graph
$G=(\mathcal{T},E)$, unconstrained $\mathcal{I}_Q$, and bounded
adjacency,
\[
S_p(f,P)
=
\max_{(u,v)\in E}
\|\phi(u)-\phi(v)\|_p,
\]
with value $0$ when $E=\emptyset$.
\end{lemma}

\begin{proof}[Proof sketch]
A policy neighbor replaces one value $u$ by $v$ with $(u,v)\in E$,
so linearity gives
$f(D)-f(D')=\phi(u)-\phi(v)$. Conversely, every policy edge can be
realized by such a neighboring pair. Taking the maximum gives the
claim.
\end{proof}

This identity is elementary and consistent with the workload view of
Haney et al.~\cite{haney2015design}; its role here is to isolate the
record-level response that the evolutionary mechanism itself is free
to redesign.

\begin{corollary}[Attribute-by-attribute decomposition]
\label{cor:attribute-decomp}
For a product domain
$\mathcal{T}_1\times\cdots\times\mathcal{T}_p$ with per-attribute
graphs $G_i=(\mathcal{T}_i,E_i)$ and the record-level product graph
$G$ of Section~\ref{sec:prelim},
\[
S_p(f,P)=\max_i S_{p,i}(f,P),
\]
where $S_{p,i}$ ranges only over edges that change attribute $i$.
Under a Gower-style distance
$d(x,y)=\sum_i w_i d_i(x_i,y_i)$,
\[
r_d(G)=\max_i w_i\,r_{d_i}(G_i).
\]
\end{corollary}

\begin{proof}[Proof sketch]
Every record-level edge changes exactly one attribute. The edge set is
therefore a disjoint union of the per-attribute edge sets, and an edge
changing attribute $i$ has distance $w_i d_i$.
\end{proof}

A data owner can consequently reason about each attribute's policy in
isolation. Protecting additional attributes increases the analytical
reach only when one of them has a larger weighted reach than those
already protected.

For example, on \Adult with weight $w=1/5$ per attribute, protecting
\texttt{age} ($\theta=7$) gives reach $0.0192$ and \texttt{hours.per.week}
($\theta=8$) gives $0.0163$; protecting both keeps the reach at $0.0192$.
Adding the binary \texttt{income} label raises it to its full weight
$0.2$, about ten times larger (Section~\ref{sec:exp-policy}).

\subsection{Hard-Vote Sensitivity: A Step Function}
\label{sec:theory-hardvote}

\begin{corollary}[Hard-vote sensitivity is set by boundary crossings]
\label{cor:indicator}
For the hard vote $\phi(u)=e_{\mathrm{nn}(u)}$ and any policy graph
$G=(\mathcal{T},E)$,
\[
S_2(f,P)=
\begin{cases}
\sqrt2,
&
\exists (u,v)\in E:
\mathrm{nn}(u)\ne\mathrm{nn}(v),\\
0,
&
\text{otherwise}.
\end{cases}
\]
\end{corollary}

\begin{proof}[Proof sketch]
If the nearest candidates agree, the two one-hot vectors are
identical. Otherwise their difference has one $+1$ and one $-1$
entry and hence $\ell_2$ norm $\sqrt2$. Lemma~\ref{lem:policy-sensitivity}
takes the maximum over policy edges.
\end{proof}

A graph whose edges all remain inside the same Voronoi cells therefore
gives sensitivity $0$. What the hard vote lacks is any graded
dependence on policy distance: its response jumps directly from
$0$ to $\sqrt2$ at the first protected boundary crossing. In
particular, once an edge crosses a boundary, shortening that edge does
not reduce sensitivity.

The crossing structure is also population-dependent. Define
\[
r_{C_t}
=
\min_{\mathrm{nn}(u)\ne\mathrm{nn}(v)}
d(u,v),
\]
the shortest substitution that crosses a current decision boundary.
As $C_t$ changes through resampling and mutation, so can $r_{C_t}$.
A fixed threshold policy can therefore yield sensitivity $0$ in one
round and $\sqrt2$ in another.

The same step behavior holds under other adjacency notions
(Table~\ref{tab:hard-vote}). An add/remove edge $(u,\bot)$ makes one
unit vote appear or disappear, a change of $1$ in both $\ell_1$ and
$\ell_2$. If public constraints force up to $k$ aligned record changes,
the triangle inequality gives at most $2k$ in $\ell_1$ and $k\sqrt2$
in $\ell_2$. In every case the response depends on whether a boundary
is crossed, not on how far the substitution moves.

\begin{table}[t]
\centering
\small
\caption{Hard-vote sensitivity by adjacency notion when a change
crosses a decision boundary.}
\label{tab:hard-vote}
\resizebox{\columnwidth}{!}{%
\begin{tabular}{lccc}
\toprule
Adjacency & effect of one change & $\ell_1$ & $\ell_2$ \\
\midrule
add/remove (edge to $\bot$)
& one vote appears or disappears
& $1$ & $1$ \\

replacement, vote moves
& vote moves from bin $j$ to $k$
& $2$ & $\sqrt2$ \\

replacement, same bin
& vote does not move
& $0$ & $0$ \\

$k$-row constrained neighbor
& up to $k$ aligned moves
& $\le 2k$ & $\le k\sqrt2$ \\
\bottomrule
\end{tabular}
}
\end{table}

\subsection{A Tight Soft-Vote Sensitivity Bound}
\label{sec:theory-softmax}

The soft vote replaces this discontinuous response by a smooth one,
but a generic softmax Lipschitz argument is not enough for our
purpose. With logits
\[
z(x)_j=-d(x,c_j)/\tau,
\]
the triangle inequality gives, for every candidate $j$,
\[
|z(x)_j-z(y)_j|
\le
\frac{d(x,y)}{\tau}.
\]
A direct $\ell_2$ bound on the logit perturbation can therefore scale
as $\sqrt m\,d(x,y)/\tau$, introducing an undesirable dependence on
the candidate count.

The key observation is that softmax is invariant to adding the same
constant to every logit. Hence only the \emph{relative} displacement
of the logits matters. Writing
$a=z(y)-z(x)$,
\[
\mathrm{osc}(a)
:=
\max_j a_j-\min_j a_j
\le
\frac{2d(x,y)}{\tau},
\]
independent of $m$. The next lemma converts this oscillation bound
into a sharp change in the resulting probability vector.

\begin{lemma}[Total variation under bounded log-tilting;
cf.~\cite{cohen2023hyperbolic}]
\label{lem:tv-tilt}
Let $p$ be a distribution on $\{1,\dots,m\}$ and
\[
q_j
=
\frac{p_j e^{a_j}}
{\sum_k p_k e^{a_k}}
\]
with $\mathrm{osc}(a)\le s$. Then
\[
\|p-q\|_1
\le
2\tanh(s/4),
\]
and the bound is attained by a suitable two-point distribution.
\end{lemma}

\begin{proof}[Proof idea]
Because adding a constant to every $a_j$ leaves $q$ unchanged, center
the perturbation so that $a_j\in[-s/2,s/2]$. Let
$W_j=e^{a_j}\in[L,U]$ with $L=e^{-s/2}$ and $U=e^{s/2}$, and let
$Z=\sum_j p_jW_j$. Then
\[
\|p-q\|_1
=
\frac{2}{Z}
\mathbb{E}_p[(W-Z)^+].
\]
Convexity of $w\mapsto(w-Z)^+$ bounds the expectation by the chord
through $L$ and $U$:
\[
\mathbb{E}_p[(W-Z)^+]
\le
\frac{(Z-L)(U-Z)}{U-L}.
\]
Substitution gives
\[
\|p-q\|_1
\le
\frac{2}{U-L}
\left[
(U+L)-\left(Z+\frac1Z\right)
\right].
\]
Since $LU=1$, AM--GM shows that the right-hand side is maximized at
$Z=1$, where it equals $2\tanh(s/4)$. The complete equality
conditions and derivation are in \appref{app:proofs}.
\end{proof}

We now apply this result to the vote.

\begin{corollary}[Sensitivity bound for \GraphSmooth]
\label{cor:softmax}
For the soft-vote histogram $H_t(D;C_t)$ and any policy graph $G$,
\[
S_2(H_t,P)
\le
\sqrt2\,
\tanh\!\left(
\frac{r_d(G)}{2\tau}
\right),
\]
independent of $m=|C_t|$. For a threshold graph $G_\theta$, this is
\[
S_2(H_t,P_\theta)
\le
\sqrt2\tanh(\theta/2\tau).
\]
\end{corollary}

\begin{proof}
By Lemma~\ref{lem:policy-sensitivity}, it is enough to bound
$\|v_{C_t}(x)-v_{C_t}(y)\|_2$ for one policy edge $(x,y)$. The vote
at $y$ is an exponential tilt of the vote at $x$, and the corresponding
logit perturbation has oscillation at most $2d(x,y)/\tau$. Hence
Lemma~\ref{lem:tv-tilt} gives
\[
\|v_{C_t}(x)-v_{C_t}(y)\|_1
\le
2\tanh\!\left(\frac{d(x,y)}{2\tau}\right).
\]
The difference of two probability vectors sums to zero, and any
zero-sum vector $z$ satisfies
$\|z\|_2\le\|z\|_1/\sqrt2$. Therefore
\[
\|v_{C_t}(x)-v_{C_t}(y)\|_2
\le
\sqrt2
\tanh\!\left(\frac{d(x,y)}{2\tau}\right).
\]
Maximizing over policy edges and using
$d(x,y)\le r_d(G)$ proves the claim.
\end{proof}

The candidate count disappears because the proof controls the
\emph{oscillation} of the logit perturbation rather than its
$\ell_2$ norm. This is also what lets the bound saturate at the
simplex diameter instead of growing without limit.

\begin{proposition}[The bound is tight]
\label{prop:lower-bound}
For every $\tau>0$, $\theta>0$, and $m\ge2$, there exist a finite
domain $\mathcal{T}\subset\mathbb{R}$, $m$ candidates, and a
threshold-$\theta$ policy such that
\[
S_2(H_t,P_\theta)
=
\sqrt2\tanh(\theta/2\tau)
\]
for $m=2$, and for any $\varepsilon>0$,
\[
S_2(H_t,P_\theta)
\ge
\sqrt2\tanh(\theta/2\tau)-\varepsilon
\]
for $m>2$.
\end{proposition}

\begin{proof}[Proof sketch]
Place two candidates at $0$ and $L$ and a protected pair
\[
u=(L-\theta)/2,
\qquad
v=(L+\theta)/2
\]
symmetrically around their midpoint. Their distance is $\theta$.
With two candidates, the soft vote reduces to a logistic function;
the two probability vectors differ in $\ell_2$ norm by exactly
$\sqrt2\tanh(\theta/2\tau)$. For $m>2$, place the remaining
candidates arbitrarily far away so that their softmax mass tends to
zero. The two-candidate construction is then approached arbitrarily
closely. The full calculation is in \appref{app:proofs}.
\end{proof}

No uniformly smaller bound depending only on reach and temperature
is therefore possible. For small $r_d(G)/\tau$,
\[
\sqrt2\tanh\!\left(\frac{r_d(G)}{2\tau}\right)
\approx
\frac{r_d(G)}{\sqrt2\,\tau},
\]
while for large $r_d(G)/\tau$ it saturates at $\sqrt2$. The bound is
worst-case over candidate populations; the exact sensitivity for a
realized $C_t$ can be substantially smaller. Section~\ref{sec:method-algorithm}
explains why we nevertheless calibrate to the closed form.

\begin{remark}[The soft bound stays below the hard vote's worst case]
\label{rem:strict-gap}
For $\tau>0$ and finite reach, the bound is strictly below $\sqrt2$.
This compares it with the hard vote's \emph{worst case}, not with its
exact policy-specific sensitivity: if no protected edge crosses a
decision boundary, Hard has sensitivity $0$. On an attribute with
only one nonzero distance, such as a binary or flat categorical
attribute, changing the graph cannot create finer distance levels;
only $\tau$ moves the bound
(Section~\ref{sec:method-categorical}).
\end{remark}

\subsection{Temperature and the Evolutionary Signal}
\label{sec:theory-convergence}

Corollary~\ref{cor:softmax} describes what a larger temperature buys
in sensitivity, but not what it costs. \texttt{ResampleMutate} needs
the released histogram to distinguish well-supported candidates from
poorly supported ones. A nearly uniform vote carries little such
signal even if it requires little noise.

\begin{proposition}[Temperature trades discrimination for sensitivity]
\label{prop:temperature-tradeoff}
Fix $C_t=\{c_1,\dots,c_m\}$.
\begin{enumerate}[leftmargin=1.4em,label=(\alph*)]
\item For $x$ with margin
\[
\gamma(x)
=
\min_{j\ne\mathrm{nn}(x)}
d(x,c_j)
-
d(x,c_{\mathrm{nn}(x)})
>0,
\]
\[
1-v_{C_t}(x)[\mathrm{nn}(x)]
\le
(m-1)e^{-\gamma(x)/\tau}.
\]
Thus the soft vote recovers the hard vote exponentially fast as
$\tau\to0^+$.

\item With
$\mathrm{diam}(C_t)=\max_{j,k}d(c_j,c_k)$, for every $x$,
\[
\left\|
v_{C_t}(x)-\frac1m\mathds{1}
\right\|_1
\le
2\tanh\!\left(
\frac{\mathrm{diam}(C_t)}{4\tau}
\right)
\le
\frac{\mathrm{diam}(C_t)}{2\tau}.
\]
Thus the vote approaches uniform at rate $O(1/\tau)$, governed by the
candidate population rather than the policy.
\end{enumerate}
\end{proposition}

\begin{proof}[Proof sketch]
For (a), divide the softmax numerator and denominator by the nearest
candidate's term and bound every remaining exponential by
$e^{-\gamma(x)/\tau}$. For (b), apply
Lemma~\ref{lem:tv-tilt} to the uniform distribution with tilt
$a_j=-d(x,c_j)/\tau$. The triangle inequality bounds the oscillation
of these logits by $\mathrm{diam}(C_t)/\tau$.
\end{proof}

The two limits expose the trade-off. Small $\tau$ preserves a sharp
nearest-neighbor signal but approaches the hard vote's sensitivity.
Large $\tau$ reduces sensitivity but drives the vote toward uniform.
At any fixed positive $\tau$, the noiseless soft histogram also
differs from the hard histogram, so the two methods can select toward
different populations even when privacy noise becomes negligible.
This is consistent with the empirical reversal at weaker privacy
budgets in Section~\ref{sec:experiments}.

\begin{remark}[No optimal temperature from the bounds]
\label{rem:no-optimal-tau}
The sensitivity term in Corollary~\ref{cor:softmax} depends only on
the public policy, distance, and temperature. The utility cost of
smoothing depends on the private data's margins, the candidate
trajectory, and the privacy budget. The theoretical bounds therefore
do not identify an optimal $\tau$. Selecting it directly on private
data would itself require privacy accounting, so we instead use the
public-data pilot of Section~\ref{sec:experiments}.
\end{remark}

Proposition~\ref{prop:temperature-tradeoff} is deliberately not a
convergence theorem. It characterizes the information retained by a
single vote at a fixed candidate population. Whether the full
population sequence converges under \texttt{ResampleMutate} depends
also on mutation, resampling, and composition. The convergence
analysis of the original hard-vote \PE has itself required subsequent
revision~\cite{gonzalez2025peconverges}. We therefore leave
convergence under soft voting to future work; the sensitivity results
above hold independently at every fixed round.

\subsection{Floors on Sensitivity}
\label{sec:bounded-unbounded}

Lemma~\ref{lem:policy-sensitivity} takes a maximum over all protected
edges. An edge whose response does not shrink with the policy reach
therefore imposes a floor on the achievable sensitivity reduction.

\begin{remark}[Sensitivity floors]
\label{rem:participation}
Two kinds of policy edges impose floors that no reach-based
calibration can lower.
\begin{enumerate}[leftmargin=1.4em,label=(\alph*)]
\item \emph{Participation edges.} An edge $(u,\bot)$ is not
controlled by the substitution distance $d$. With $\phi(\bot)=0$,
it contributes $\|\phi(u)\|$, so if participation edges are included
for all values, then $S(f,P)\ge \max_u\|\phi(u)\|$. This equals $1$
for the hard vote and is at most $1$ for the soft vote, regardless
of the substitution graph.

\item \emph{Maximal-distance attributes.} Under a Gower-style
distance, if attribute $i$ contains a protected pair at its maximum
distance $\mathrm{diam}(\mathcal{T}_i)$, then
Corollary~\ref{cor:attribute-decomp} gives
$r_d(G)\ge w_i\,\mathrm{diam}(\mathcal{T}_i)$, and the analytical
soft-vote bound is at least
\[
\sqrt2\tanh\!\left(\frac{w_i\,\mathrm{diam}(\mathcal{T}_i)}{2\tau}\right).
\]
Complete graphs have this property. Under a flat categorical Gower
term, so does any block with two or more distinct values, and a
binary attribute admits no policy between empty and complete.
\end{enumerate}
\end{remark}

These floors do not imply that one policy is intrinsically worse than
another for a fixed response function. They identify policy edges
that the particular reach-dependent reduction of
Corollary~\ref{cor:softmax} cannot make smaller.

\paragraph{Scope of the analysis.}
The results in this section assume bounded replacement adjacency and
unconstrained $\mathcal{I}_Q$. The main reduction is strongest when
the protected substitutions have small reach under the same metric
used by the vote, as with narrow numeric, ordinal, or hierarchical
policies. Participation edges and maximal-distance categorical edges
can limit that reduction. The results are per-round sensitivity
statements rather than convergence guarantees, but because they hold
conditional on every realized candidate population, they compose
through the adaptive evolutionary process described in
Section~\ref{sec:method}.
\section{Experiments}
\label{sec:experiments}

We ask four questions: when and why does soft voting outperform hard
voting (Sections~\ref{sec:exp-budget}--\ref{sec:exp-mechanism})?
Can $\tau$ be chosen without private data
(Section~\ref{sec:exp-tau-selection})? Which policies preserve the
gain (Section~\ref{sec:exp-policy})? Does the sensitivity bound hold,
and what does calibration cost
(Section~\ref{sec:exp-sensitivity-summary})?
Protocols and supporting results are in \appref{app:experiments}.

\subsection{Setup}
\label{sec:exp-setup}

\mypar{Datasets.}
\Adult~\cite{becker1996adult} and \Bank~\cite{moro2014databn,moro2014bankdata} each
use five modeled columns including the label
(Table~\ref{tab:reference-sens}), with 26{,}048 and 27{,}126 private
training records. \XORDS draws $k$ independent ordinal inputs
$X_i\sim\mathrm{Unif}\{0,\dots,7\}$ and sets
$Y=B_1\oplus\cdots\oplus B_k$, where
$B_i=\mathbf{1}\{X_i\ge4\}$, following \TabPE's high-order
evaluation~\cite{tran2026tabpe}. We use $k=3$ (\XORDS-3) and
$n=12{,}000$.

\begin{table}[t]
\centering
\small
\caption{Default per-attribute policy graphs.} \vspace{-3mm}
\label{tab:reference-sens}
\begin{tabular}{lll}
\toprule
Dataset & Attribute & Graph \\
\midrule
\multirow{5}{*}{\Adult}
 & \texttt{age}            & threshold (7 years) \\
 & \texttt{hours.per.week} & threshold (8 hours) \\
 & \texttt{education.num}  & threshold (2 levels) \\
 & \texttt{workclass}      & block (4 blocks) \\
 & \texttt{income}         & complete \\
\midrule
\multirow{5}{*}{\Bank}
 & \texttt{age}       & threshold (7 years) \\
 & \texttt{balance}   & threshold (500) \\
 & \texttt{education} & threshold (1 level) \\
 & \texttt{job}       & hierarchy (4 groups) \\
 & \texttt{y}         & complete \\
\bottomrule
\end{tabular}
\end{table}

\mypar{Policies and reach.}
We use equal-weight Gower distance, $w=1/p$, over the modeled columns
($p=5$ on \Adult and \Bank, $p=4$ on \XORDS-3). Numeric attributes are
normalized by fixed domains: $[17,90]$ for \Adult age, $[18,100]$
for \Bank age, and $\{0,\dots,7\}$ for $X_i$. A threshold $\theta$
then contributes reach $w\theta/\text{width}$, and multi-attribute
policies take the maximum (Corollary~\ref{cor:attribute-decomp}).
Except in Section~\ref{sec:exp-policy}, we protect one attribute:
age with $\theta=7$ on \Adult and \Bank (reach $0.0192$ and $0.0171$)
and $X_1$ with one ordinal step on \XORDS-3 (reach $1/28$). Other modeled
attributes lie outside the policy. Section~\ref{sec:exp-policy}
protects combinations up to all five columns; \appref{app:exp-policy}
gives the rationale.

\mypar{Methods.}
All methods use the same implementation of Algorithm~\ref{alg:unified},
initial population, selection/mutation schedule, budget split, and
paired randomness. \Hard is the one-hot \BFEvo vote. In every checked
round of the fixed-temperature runs, a protected edge crossed a
nearest-candidate boundary, so its exact sensitivity was $\sqrt2$,
matching \TabPE's calibration; broader-policy runs use $\sqrt2$
directly. \GraphSmooth uses Corollary~\ref{cor:softmax}. We report
its noise ratio%
\[
q=\sigma_S/\sigma_H=\tanh(r_d(G)/2\tau).
\]
A standard-DP soft vote has full reach and, for the temperatures used
in Section~\ref{sec:exp-mechanism}, effectively $q=1$.

The noise ratio has a direct budget interpretation. Because
$\sigma_t^2=\Delta_t^2/(2\rho_t)$, \GraphSmooth at budget $\rho$ adds
the same Gaussian noise that \Hard would add at budget $\rho/q^2$. For
the default \Adult policy (\texttt{age}, $\theta=7$, reach $0.0192$)
at $\tau=0.03$, $q=0.31$, so \GraphSmooth's noise matches that of
\Hard with roughly $10\times$ the privacy budget. This equivalence
concerns noise only: the soft vote also changes the released
histogram, and Section~\ref{sec:exp-mechanism} separates the two
effects.

\mypar{Measures and protocol.}
The primary error is mean $\ell_1$ error over all two-way marginals
(\PairwiseErr) on \Adult and \Bank, and $\ell_1$ error over the
16-cell joint distribution of $(B_1,B_2,B_3,Y)$ (\HighOrderErr) on
\XORDS-3. Lower is better. Each release has 512 records over $T=8$ rounds
with equal budget allocation. We also report Blowfish $\epsilon$ at
$\delta=10^{-5}$ via Eq.~\eqref{eq:zcdp-conversion}; for example,
$\rho=10^{-4},6\times10^{-4},3\times10^{-3}$ give
$\epsilon\approx0.07,0.17,0.37$. Intervals are pointwise $95\%$
Student-$t$ intervals over paired runs. Best-of-sweep settings are
descriptive and uncorrected for selection
(\appref{app:exp-protocol}).

\subsection{Results}
\label{sec:exp-results}

\begin{figure*}[t]
 \centering
 \begin{subfigure}[b]{0.32\textwidth}
  \centering
  \includegraphics[width=\textwidth]{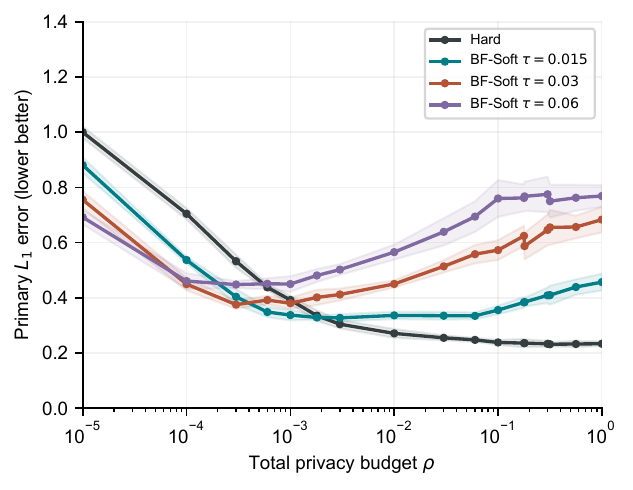}
  \caption{\Adult}
  \label{fig:crossover-adult}
 \end{subfigure}\hfill
 \begin{subfigure}[b]{0.32\textwidth}
  \centering
  \includegraphics[width=\textwidth]{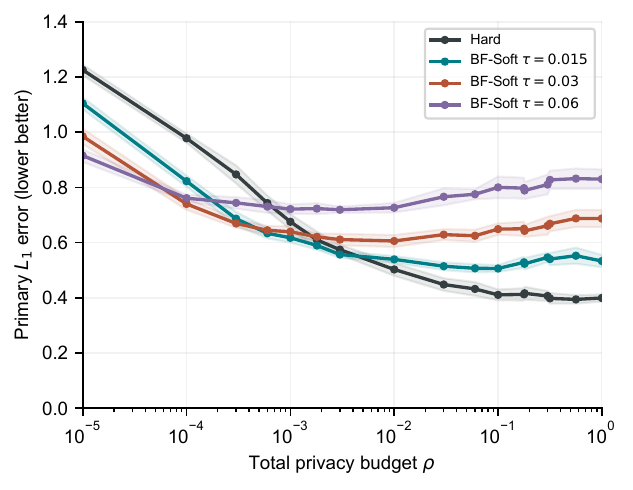}
  \caption{\Bank}
  \label{fig:crossover-bank}
 \end{subfigure}\hfill
 \begin{subfigure}[b]{0.32\textwidth}
  \centering
  \includegraphics[width=\textwidth]{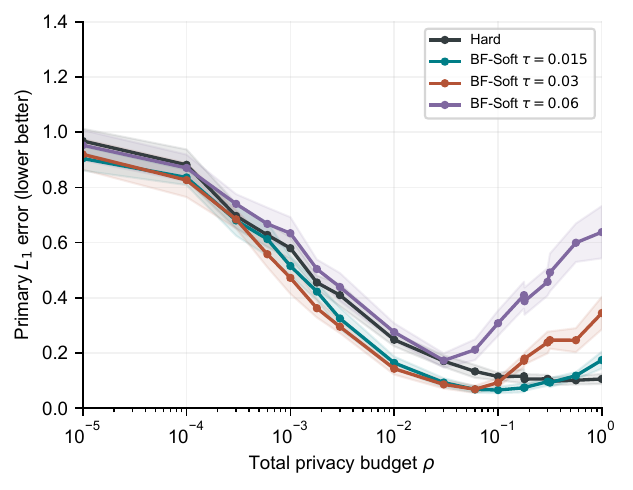}
  \caption{\XORDS-3}
  \label{fig:crossover-xor3}
 \end{subfigure}
 \vspace{-1mm}
 \caption{Error across privacy budgets at three fixed temperatures and
 default reach.}
 \label{fig:exp-crossover}
\end{figure*}

\subsubsection{Soft voting helps at strong privacy and reverses at weak privacy}
\label{sec:exp-budget}

\begin{table}[t]
\centering
\small
\caption{Noise ratio $q$, interpolated crossing budget $\rho_\times$
(linear in $\log\rho$ between neighboring measured budgets), and
corresponding $\epsilon$ at $\delta=10^{-5}$.}
\label{tab:crossings}
\setlength{\tabcolsep}{4pt}
\begin{tabular}{llccc}
\toprule
Dataset & $\tau$ & $q$ & $\rho_\times$ & $\epsilon$ \\
\midrule
\multirow{3}{*}{\Adult}
 & 0.015 & 0.56 & $2.0\times10^{-3}$ & 0.31 \\
 & 0.03  & 0.31 & $1.1\times10^{-3}$ & 0.23 \\
 & 0.06  & 0.16 & $5.5\times10^{-4}$ & 0.16 \\
\midrule
\multirow{3}{*}{\Bank}
 & 0.015 & 0.51 & $4.4\times10^{-3}$ & 0.46 \\
 & 0.03  & 0.28 & $1.6\times10^{-3}$ & 0.27 \\
 & 0.06  & 0.14 & $6.7\times10^{-4}$ & 0.18 \\
\midrule
\multirow{3}{*}{\XORDS-3}
 & 0.015 & 0.83 & $0.36$              & 4.5 \\
 & 0.03  & 0.53 & $0.12$              & 2.5 \\
 & 0.06  & 0.29 & $1.3\times10^{-4}$ & 0.08 \\
\bottomrule
\end{tabular}
\end{table}

Figure~\ref{fig:exp-crossover} compares \Hard and \GraphSmooth at
$\tau\in\{0.015,0.03,0.06\}$ over 17 budgets from $10^{-5}$ to $1$,
with twelve paired runs. Larger $\tau$ saves more noise:
$q=0.56,0.31,0.16$ on \Adult; $0.51,0.28,0.14$ on \Bank; and
$0.83,0.53,0.29$ on \XORDS-3. In all nine comparisons, \GraphSmooth has
lower mean error at $\rho=10^{-5}$ and higher error at $\rho=1$, with
one sign change between them.

On \Adult and \Bank, interpolated crossings occur at
$\rho\approx5.5\times10^{-4}$--$4.4\times10^{-3}$
($\epsilon\approx0.16$--$0.46$) and shift downward as $\tau$ grows
(Table~\ref{tab:crossings}). Thus more smoothing helps when privacy
noise dominates but hurts once smoothing cost dominates, consistent
with Section~\ref{sec:theory-convergence}. \XORDS-3 crossings range more
widely, from $\rho\approx1.3\times10^{-4}$ to $0.36$.

Training size has a similar effect. On \XORDS-3 at $\rho=0.003$ and
$\tau=0.03$, mean error reduction is $22.8\%$ at $n=12{,}000$,
$37.4\%$ at $48{,}000$, and $-89.2\%$ at $100{,}000$
(\appref{app:exp-trainsize}).

\begin{figure*}[t]
 \centering
 \begin{subfigure}[b]{0.32\textwidth}
  \centering
  \includegraphics[width=\textwidth]{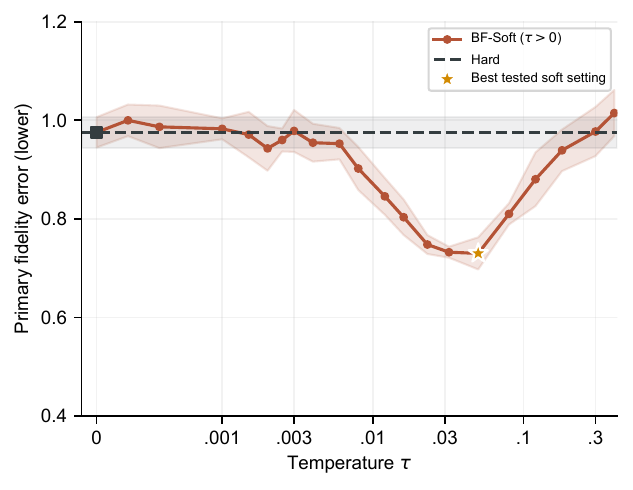}
  \caption{$\rho=10^{-4}$}
  \label{fig:interior-bank-a}
 \end{subfigure}\hfill
 \begin{subfigure}[b]{0.32\textwidth}
  \centering
  \includegraphics[width=\textwidth]{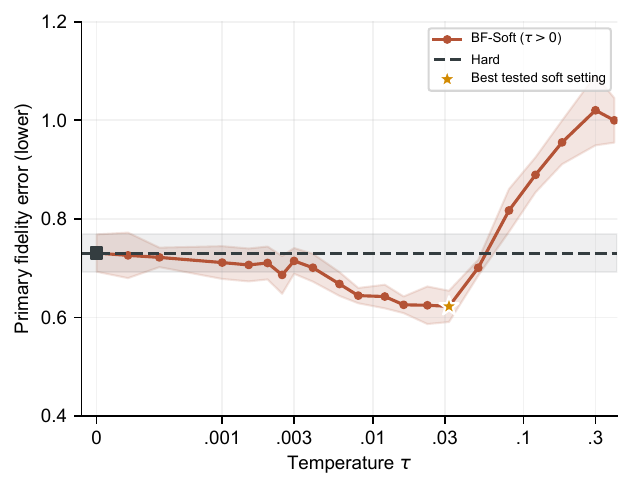}
  \caption{$\rho=6\times10^{-4}$}
  \label{fig:interior-bank-b}
 \end{subfigure}\hfill
 \begin{subfigure}[b]{0.32\textwidth}
  \centering
  \includegraphics[width=\textwidth]{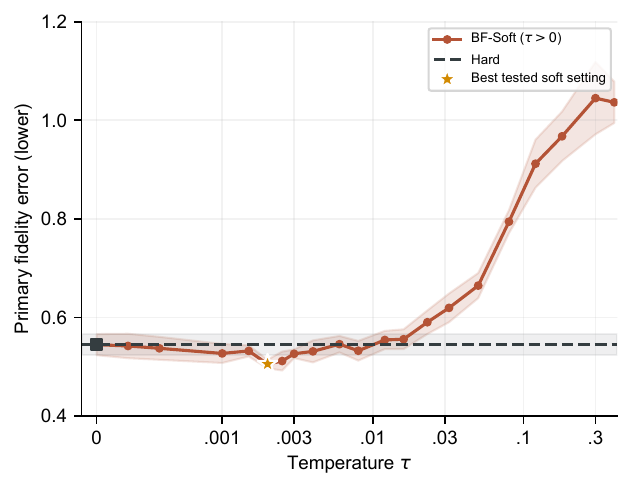}
  \caption{$\rho=3\times10^{-3}$}
  \label{fig:interior-bank-c}
 \end{subfigure} \vspace{-1mm}
 \caption{Temperature sweeps on \Bank.}
 \label{fig:exp-interior}
\end{figure*}

\subsubsection{Temperature sweeps and the source of the gain}
\label{sec:exp-interior}

Figure~\ref{fig:exp-interior} sweeps 20 temperatures from $0.00025$
to $0.4$ on \Bank; \appref{app:exp-interior} gives \Adult and \XORDS-3.
In all nine dataset--budget panels, the lowest sampled
\GraphSmooth mean occurs at an interior $\tau\in[0.002,0.05]$ and is
below \Hard.

At $\rho=10^{-4}$ ($\epsilon\approx0.07$), the best tested settings
reduce mean error from $0.696$ to $0.427$ on \Adult and $0.976$ to
$0.730$ on \Bank ($38.6\%$ and $25.2\%$). The gain is not confined to
one setting: \Adult at $10^{-4}$ and \Bank at $6\times10^{-4}$ remain
below \Hard for every tested $\tau\in[0.004,0.05]$. As $\rho$ grows,
the best $\tau$ falls: at $\rho=3\times10^{-3}$
($\epsilon\approx0.37$), it is $0.006$ on \Adult and $0.002$ on
\Bank, with $9.9\%$ and $7.2\%$ lower mean error. These
best-of-sweep results are optimistic; Section~\ref{sec:exp-tau-selection}
selects $\tau$ prospectively.

The corresponding noise ratios explain the source of the gain. At
$\rho=10^{-4}$, $q\approx0.29$ on \Adult and $0.17$ on \Bank, so
most of \Hard's noise is removed. At $3\times10^{-3}$,
$q\approx0.92$ and $1.00$: little or no noise is saved, and the
remaining gain comes from the soft vote itself. A standard-DP soft
vote would have essentially the same noise there. On \XORDS-3 the gain is
smaller and less certain: $18.4\%$ at $3\times10^{-3}$, while paired
intervals include zero at $10^{-4}$ and $6\times10^{-4}$.

At weaker budgets the gain on \XORDS-3 comes from the vote itself rather
than from saved noise. Jointly varying $\tau$ and the protected width
$\theta\in\{1,2,3,4\}$ gives largest mean relative gains of $39.8\%$,
$40.4\%$, and $19.9\%$ at $\rho=0.03$, $0.18$, and $1$
($\epsilon\approx1.2$, $3.1$, and $7.8$), at settings where
$r_d(G)/2\tau>5$ and hence $q\approx1$ (\appref{app:exp-xor-interaction}).
A standard-DP soft vote would receive essentially the same noise there,
so on this high-order workload smoothing can improve selection even
when it saves no privacy noise.

\begin{figure*}[t]
 \centering
 \begin{subfigure}[b]{0.32\textwidth}
  \centering
  \includegraphics[width=\textwidth]{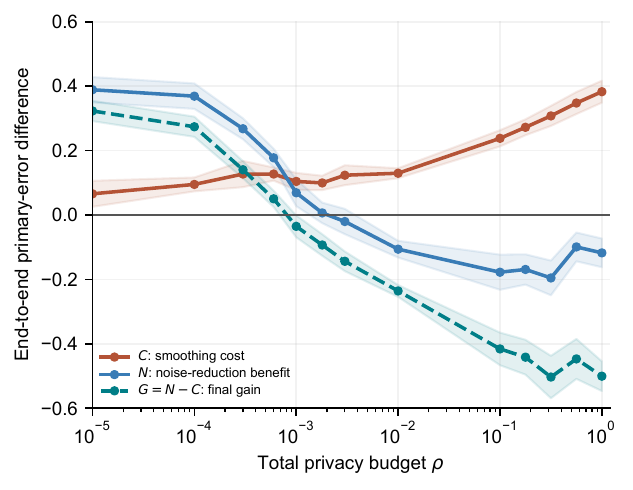}
  \caption{\Adult, $\tau=0.045$}
  \label{fig:mech-adult}
 \end{subfigure}\hfill
 \begin{subfigure}[b]{0.32\textwidth}
  \centering
  \includegraphics[width=\textwidth]{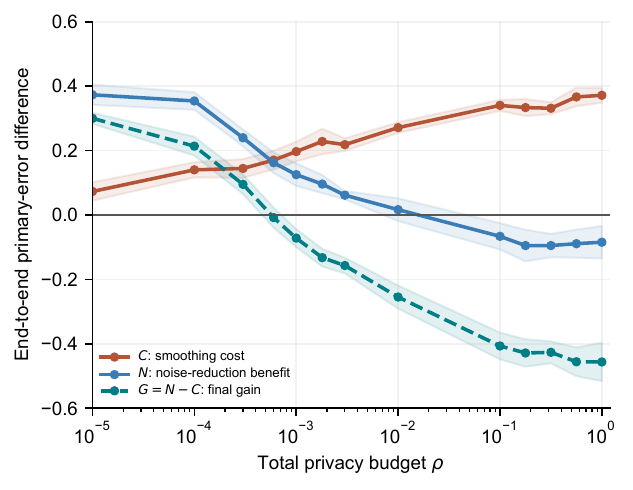}
  \caption{\Bank, $\tau=0.065$}
  \label{fig:mech-bank}
 \end{subfigure}\hfill
 \begin{subfigure}[b]{0.32\textwidth}
  \centering
  \includegraphics[width=\textwidth]{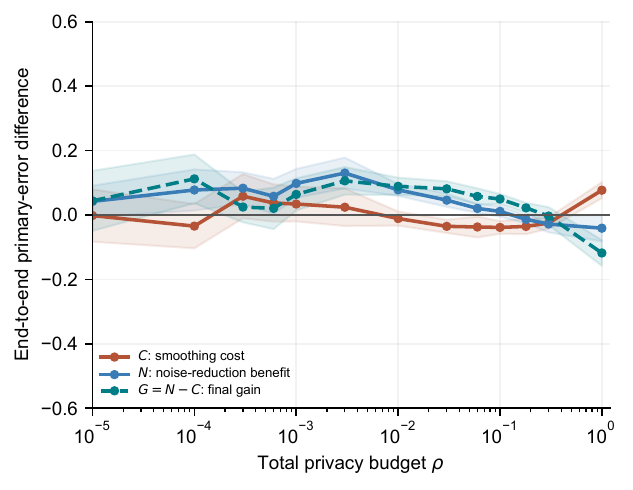}
  \caption{\XORDS-3, $\tau=0.02$}
  \label{fig:mech-xor3}
 \end{subfigure} \vspace{-1mm}
 \caption{Smoothing cost $C$, noise-reduction benefit $N$, and final
 gain $G=N-C$ (Eq.~\eqref{eq:exp-decomposition}).}
 \label{fig:exp-mechanism}
\end{figure*}

\begin{figure*}[t]
 \centering
 \begin{subfigure}[b]{0.32\textwidth}
  \centering
  \includegraphics[width=\textwidth]{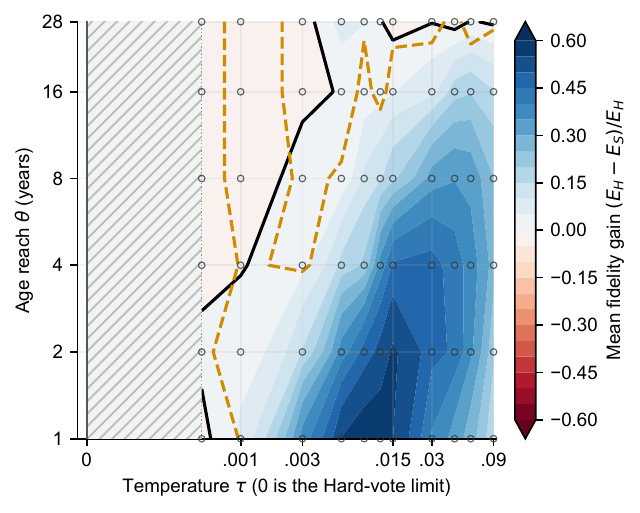}
  \caption{$\rho=10^{-4}$}
  \label{fig:forecast-adult-a}
 \end{subfigure}\hfill
 \begin{subfigure}[b]{0.32\textwidth}
  \centering
  \includegraphics[width=\textwidth]{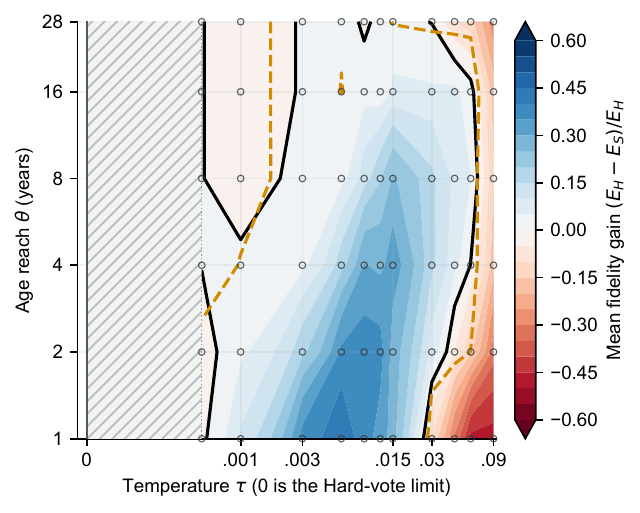}
  \caption{$\rho=6\times10^{-4}$}
  \label{fig:forecast-adult-b}
 \end{subfigure}\hfill
 \begin{subfigure}[b]{0.32\textwidth}
  \centering
  \includegraphics[width=\textwidth]{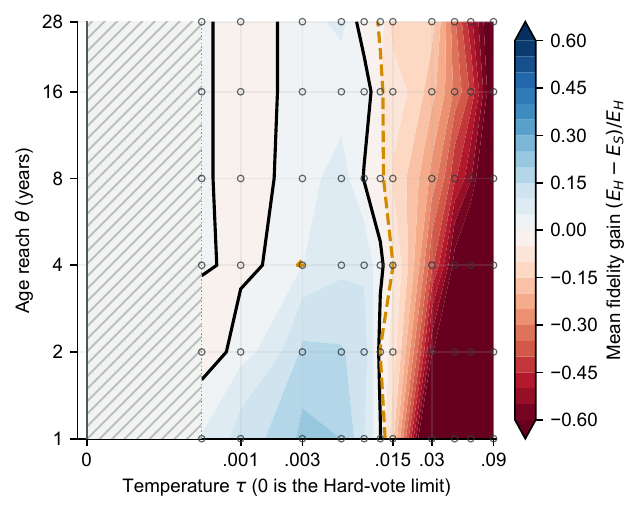}
  \caption{$\rho=3\times10^{-3}$}
  \label{fig:forecast-adult-c}
 \end{subfigure} 
 \caption{Relative gain on \Adult over temperature and \texttt{age}
 reach. Black: observed zero-gain contour. Gold: public-pilot forecast,
 frozen before the target runs.}
 \label{fig:exp-tau-selection}
\end{figure*}

\subsubsection{Noise savings versus smoothing cost}
\label{sec:exp-mechanism}

Let $E_H$ and $E_S$ be the errors of \Hard and \GraphSmooth, and
$E_M$ the error of the soft vote with \Hard's noise. Define
\begin{equation}
 C=E_M-E_H,\qquad
 N=E_M-E_S,\qquad
 G=E_H-E_S=N-C.
 \label{eq:exp-decomposition}
\end{equation}
Here $C$ is the smoothing cost at matched noise, $N$ the benefit of
policy-specific noise reduction, and \GraphSmooth wins iff $N>C$.
At these temperatures, $E_M$ is also the error of the soft vote under
standard DP.

Figure~\ref{fig:exp-mechanism} reports these terms at one temperature
per dataset over twelve paired runs. On \Adult and \Bank, $C>0$ at
every budget, so standard-DP soft voting is worse than \Hard; the
strong-privacy gain comes entirely from policy-induced noise
reduction. As $\rho$ grows, $N$ decreases while $C$ grows, producing
the crossover. On \XORDS-3 both effects are small; $N$ first rises and
then falls, consistent with severe noise and clipping limiting useful
selection at the smallest budgets, although the decomposition does
not isolate clipping.

The noise ratio alone does not determine fidelity. Holding $q$ fixed
within a dataset while varying reach and temperature together
($q=0.18$ on \Adult and \Bank with $\theta\in\{2,4,8\}$; $q=0.95$ on
\XORDS-3 with $\theta\in\{1,2,4\}$) still yields different error
curves across 25 budgets (\appref{app:exp-equal-noise}). The temperature
changes the vote itself, and with it the evolutionary trajectory, so
equal noise does not imply equal utility. This is the smoothing cost
$C$ seen from another angle: two settings with the same noise can
differ in how much selection signal their votes retain.

\subsubsection{Selecting $\tau$ without private data}
\label{sec:exp-tau-selection}

The beneficial region depends jointly on $\tau$, reach, and budget
and is not determined by the bound alone
(Remark~\ref{rem:no-optimal-tau}). We therefore test a public pilot.
Each pilot uses 3{,}000 records disjoint from private training data;
for \Adult and \Bank these are held-out validation records, while
\XORDS-3 uses separately generated data. We rescale the pilot budget to
match Gaussian noise in histogram proportions:
\begin{equation}
 \rho_p=\rho_t\left(\frac{n_t}{n_p}\right)^2,
 \qquad
 \frac{\sigma_p}{n_p}=\frac{\sigma_t}{n_t},
 \label{eq:exp-pilot-scaling}
\end{equation}
for pilot and target sizes $n_p,n_t$. This matches normalized noise
when sensitivity, output size, and round schedule agree, although
finite-sample vote distributions and adaptive trajectories may differ.

For \Adult, the pilot has $n_p=3{,}000$ records and the private table
$n_t=26{,}048$, so Eq.~\eqref{eq:exp-pilot-scaling} runs the pilot at
about $75\times$ the target budget. The pilot's noise is then the same
fraction of its histogram mass as the target's, even though the pilot
sees far fewer records.

On the original 210-setting \Adult grid, the scaled pilot predicts the
sign of mean gain in 190 cases ($90.5\%$), versus 134 ($63.8\%$)
without scaling. Sign accuracy overstates performance near the
zero-gain boundary and does not measure selected-setting utility.
Figure~\ref{fig:exp-tau-selection} shows that the pilot contour tracks
the target well at the two smaller budgets and less closely at
$3\times10^{-3}$ (\appref{app:exp-forecast}).
Section~\ref{sec:exp-policy} evaluates pilot-selected settings on
\Bank.

These pilots come from the same benchmark distribution as the target;
performance under independently collected or shifted public data is
untested. A public-only synthetic table is also a natural baseline
when such data are available, but we did not evaluate it.

\begin{figure*}[t]
 \centering
 \begin{subfigure}[b]{0.48\textwidth}
  \centering
  \includegraphics[width=\textwidth]{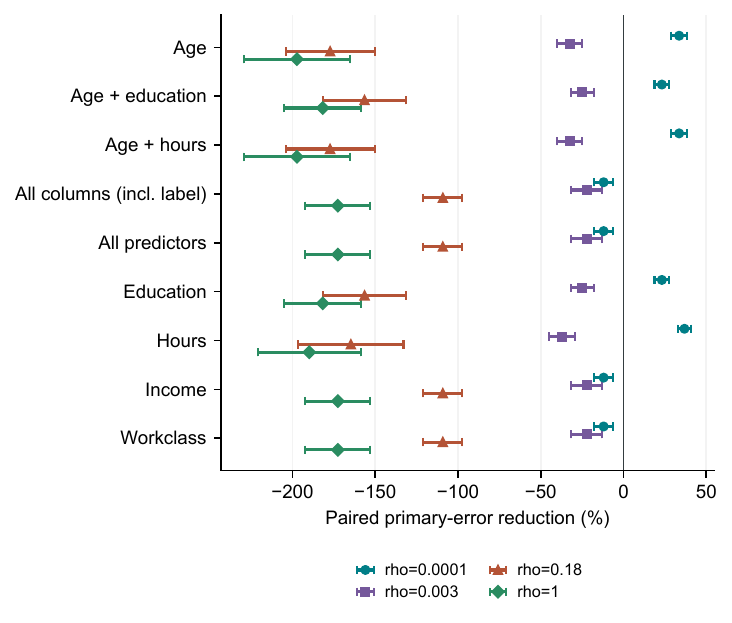}
  \caption{\Adult, fixed $\tau=0.03$}
  \label{fig:policy-adult}
 \end{subfigure}\hfill
 \begin{subfigure}[b]{0.48\textwidth}
  \centering
  \includegraphics[width=\textwidth]{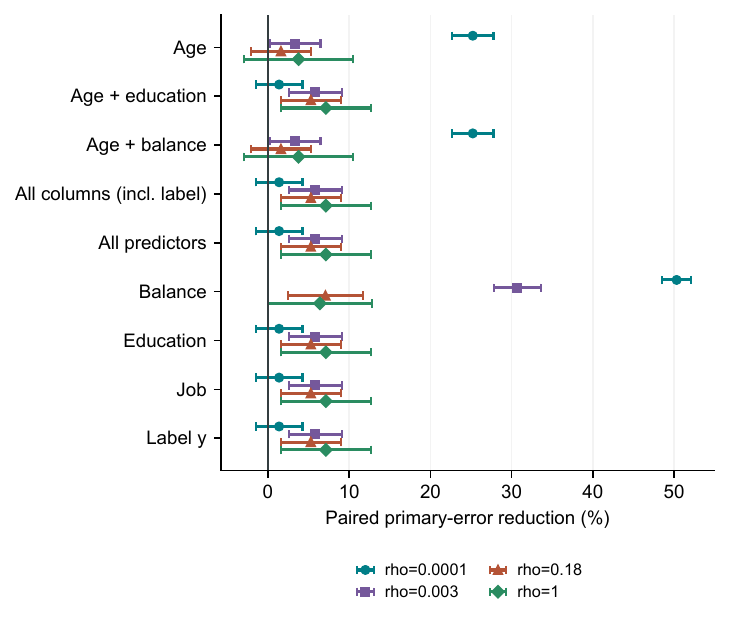}
  \caption{\Bank, method and $\tau$ selected on a public pilot}
  \label{fig:policy-bank}
 \end{subfigure} \vspace{-3mm}
 \caption{Paired error reduction relative to \Hard (positive favors
 \GraphSmooth) under single- and multi-attribute policies.}
 \label{fig:exp-policy-scope}
\end{figure*}

\begin{figure*}[t]
 \centering
 \begin{subfigure}[b]{0.32\textwidth}
  \centering
  \includegraphics[width=\textwidth]{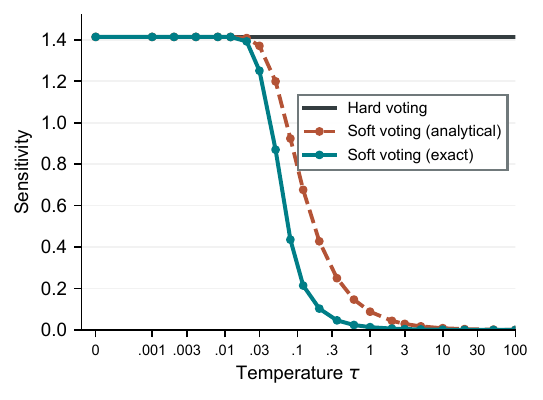}
  \caption{Sensitivity vs.\ temperature}
  \label{fig:sens-temp}
 \end{subfigure}\hfill
 \begin{subfigure}[b]{0.32\textwidth}
  \centering
  \includegraphics[width=\textwidth]{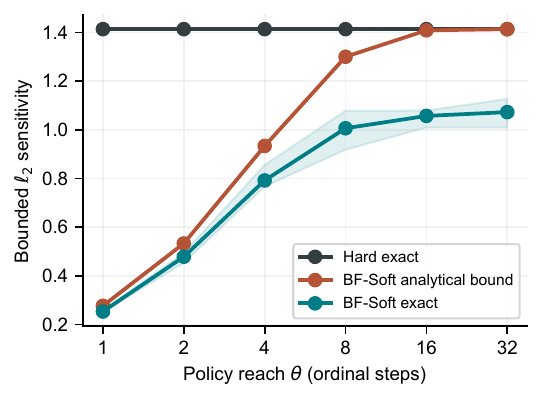}
  \caption{Sensitivity vs.\ reach}
  \label{fig:sens-reach}
 \end{subfigure}\hfill
 \begin{subfigure}[b]{0.32\textwidth}
  \centering
  \includegraphics[width=\textwidth]{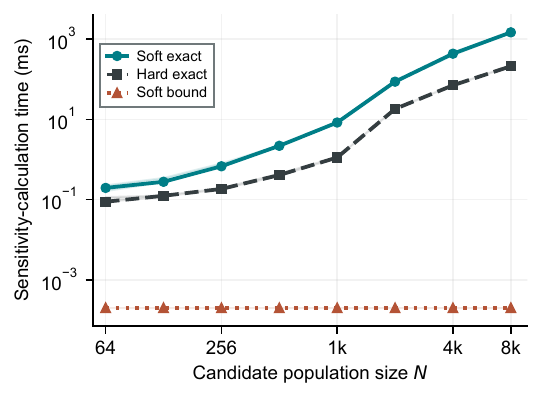}
  \caption{Calibration runtime}
  \label{fig:sens-runtime}
 \end{subfigure} \vspace{-3mm}
 \caption{Exact sensitivity against the bound of
 Corollary~\ref{cor:softmax}, and calibration cost. In (c), $N$ denotes
 the candidate population size $m$.}
 \label{fig:exp-sensitivity-summary}
\end{figure*}

\subsubsection{Which policies preserve the gain}
\label{sec:exp-policy}

Corollary~\ref{cor:attribute-decomp} and
Remark~\ref{rem:participation}(b) predict which policies can benefit.
Because reach is the maximum weighted per-attribute reach, adding
narrow numeric attributes need not change it. A flat categorical or
binary attribute raises it to that attribute's full Gower weight,
here $0.2$, making $q\approx1$ at useful temperatures.
Figure~\ref{fig:exp-policy-scope} tests this prediction on both real
datasets, with twelve paired runs per setting and Holm adjustment~\cite{holm1979simple} over 680 prespecified contrasts (\appref{app:exp-broader}).

On \Adult at $\tau=0.03$, protecting age alone reduces mean paired
error by $33.6\%$ ($[28.7,38.4]\%$) at $\rho=10^{-4}$, with
$q=0.31$. Adding \texttt{hours.per.week} leaves reach at $0.0192$,
so calibration, trajectories, and errors are unchanged. Numeric-only
policies retain positive reductions; policies including
\texttt{workclass} or \texttt{income} (reach $0.2$) do not.
Protecting all five columns gives $q=0.9975$ and $-12.0\%$
($[-17.6,-6.5]\%$), not significant after Holm adjustment
($p_{\mathrm{Holm}}=0.174$). At larger budgets, this fixed-temperature
\GraphSmooth is worse than \Hard under every policy.

On \Bank, a public pilot selects between \Hard and nine temperatures
before twelve target runs. Numeric-threshold policies retain a clear
advantage at $\rho=10^{-4}$: about $25\%$ for age and age+balance,
and $50\%$ for balance alone, whose reach is much smaller. Adding
\texttt{education}, \texttt{job}, or \texttt{y} reduces the gain below
about $10\%$. With all five columns protected, the pilot selects
$\tau=0.003$ at $\rho=0.003$, giving $5.8\%$ reduction
($[2.6,9.1]\%$), but $q=1$ shows that the gain comes from vote
smoothing rather than lower noise and it does not survive Holm
adjustment ($p_{\mathrm{Holm}}=0.613$). At $\rho=10^{-4}$,
$\tau=0.001$ gives $1.4\%$ ($[-1.5,4.3]\%$).

Thus $q$, computable from the policy before accessing private data,
largely predicts whether noise reduction remains available. Narrow
numeric policies with $q\ll1$ preserve the strong-privacy gain;
categorical or binary protection leaves little noise to save.
Graded categorical distances may recover part of this benefit
(Section~\ref{sec:method-categorical}).

\subsubsection{The calibration bound holds and is cheap}
\label{sec:exp-sensitivity-summary}

Figure~\ref{fig:exp-sensitivity-summary}\subref{fig:sens-temp}
compares exact sensitivity with the analytical bound over 1{,}320
configurations of populations, geometries, seeds, and temperatures
(\appref{app:exp-protocol}). Every measured \GraphSmooth sensitivity
lies below the bound. Both the bound and median exact sensitivity fall
with $\tau$; at $\tau\ge20$ both are below $1\%$ of $\sqrt2$. The
bound is conservative on these populations.

Increasing reach has the opposite effect
(Figure~\ref{fig:exp-sensitivity-summary}\subref{fig:sens-reach}):
both analytical and exact sensitivity increase. Every plotted policy
crosses a hard-vote boundary, so \Hard remains at $\sqrt2$, whereas
\GraphSmooth provides intermediate calibration levels.

Calibration is also cheap
(Figure~\ref{fig:exp-sensitivity-summary}\subref{fig:sens-runtime}).
With $m=8{,}192$ candidates and domain points, exact \Hard and
\GraphSmooth sensitivity take $211$\,ms and $1.48$\,s, respectively,
versus about $2\times10^{-4}$\,ms for the bound. On product domains
the gap is larger because exact calibration searches combinations of
unprotected attribute values, while the bound's cost is unchanged
(Section~\ref{sec:method-algorithm}).

\mypar{Practical guidance}
The results suggest a simple procedure for a data owner. First,
compute the noise ratio $q=\tanh(r_d(G)/2\tau)$ from the policy alone,
before touching private data. If the protected attributes are
numeric or ordinal with narrow thresholds, $q\ll1$ is achievable and
\GraphSmooth is most useful at strong privacy budgets, below the
crossings of Table~\ref{tab:crossings}
(Sections~\ref{sec:exp-budget} and~\ref{sec:exp-policy}). If the
policy protects a flat categorical or binary attribute, $q\approx1$
at useful temperatures, so little noise can be saved; there, \Hard
is the safer default unless a graded categorical distance is
available (Section~\ref{sec:method-categorical}). Second, choose
$\tau$, or fall back to \Hard, with a public-data pilot
(Section~\ref{sec:exp-tau-selection}); this spends no private budget.
Third, calibrate to the closed-form bound of
Corollary~\ref{cor:softmax}, which is computed once and remains valid
for every candidate population (Section~\ref{sec:exp-sensitivity-summary}).
\section{Conclusion and Future Work}
\label{sec:conclusion}

We asked whether grafting a Blowfish policy onto an evolutionary DP
synthesizer improves its utility, and what determines when it does.
The answer depends on the vote, not on the policy graph alone. For any
linear query, policy-specific sensitivity is a maximum over the
policy's edges (Lemma~\ref{lem:policy-sensitivity}). The hard vote's
sensitivity jumps to its worst case at the first protected boundary
crossing, so a policy helps only when no protected edge crosses a
boundary (Corollary~\ref{cor:indicator}). \GraphSmooth{} replaces it
with a soft vote whose sensitivity has a closed-form bound in the
policy reach and temperature, attained with two candidates and
approached with more (Corollary~\ref{cor:softmax},
Proposition~\ref{prop:lower-bound}). Participation edges and
maximal-distance protected attributes impose floors that this
reach-dependent reduction cannot remove
(Remark~\ref{rem:participation}).

A policy graph therefore reduces noise only when the released statistic
responds to the distances the policy encodes. The hard vote does not;
a soft vote does.

Temperature creates the corresponding utility trade-off: increasing
it lowers sensitivity but weakens the vote's ability to discriminate
among candidates, and our bounds do not identify an optimal value
(Remark~\ref{rem:no-optimal-tau}). On real and synthetic data, we
observe a beneficial region whose location shifts with policy reach
and privacy budget. A public-data pilot predicts the sign of the gain
on the original \Adult grid with $90.5\%$ accuracy without spending
private budget (Section~\ref{sec:exp-tau-selection}).

Our analysis covers linear queries under bounded, unconstrained
adjacency. The experiments consider both single- and multi-attribute
policies, with the strongest policy-induced noise reductions under
narrow numeric policies. We give no convergence analysis for
\GraphSmooth{} under \texttt{ResampleMutate}, and the public pilot is
validated only when pilot and target share a distribution.

A convergence analysis for \GraphSmooth{} and a regret guarantee for
public-pilot temperature selection would address these limitations.
The two relaxations of \PE's vote---\Wasp's rank-based sharing and our
distance-based smoothing---could also be combined. More broadly, the
sensitivity identity applies to any linear query under a graph- or
metric-based relaxation of DP~\cite{kifer2014pufferfish,chatzikokolakis2013broadening,he2014blowfish}.
Whether other synthesis paradigms, such as marginal-based mechanisms
or deep generative models trained under DP-SGD, admit a per-record
response that tracks a policy's notion of closeness remains open.

\newpage

\section*{Artifacts}

All artifacts needed to reproduce the results in this paper are publicly
available without registration or access monitoring at~\cite{blowfish-evo-synth-repo}.
The repository contains:
(i)~the implementation of Algorithm~\ref{alg:unified} with the three voting and calibration
variants (\texttt{Tab-PE}, \texttt{Hard}, and \texttt{BF-Soft});
(ii)~the policy-graph specifications of Table~\ref{tab:reference-sens} and the reach computation of
Corollary~\ref{cor:attribute-decomp};
(iii)~the generator for the synthetic \texttt{XOR} datasets;
(iv)~scripts that run each experiment in Section~\ref{sec:experiments} and regenerate
Figures~\ref{fig:exp-crossover}--\ref{fig:exp-sensitivity-summary}, including the public-pilot temperature selection and the
exact-sensitivity checks; and
(v)~a \texttt{README} describing the software environment, dependencies,
and the command for each experiment.
The \texttt{Adult}~\cite{becker1996adult} and
\texttt{Bank}~\cite{moro2014bankdata} datasets are publicly available from
the UCI Machine Learning Repository; the repository includes the
preprocessing scripts and the training/validation splits used for the
private data and the public pilots.
The extended version of the paper, with all appendices and full
proofs, is available at~\cite{extendedversion}.

\balance
\bibliographystyle{ACM-Reference-Format}
\bibliography{references}

\ifextended
\appendix
\appendix

\section{Extended Related-Work Differentiation}
\label{app:related}

Section~\ref{sec:related} gives the distinctions needed to position
our contribution. This appendix expands three comparisons that are
particularly close to the construction or analysis in this paper.

\paragraph{\Wasp{}'s rank-based weighting remains piecewise constant.}
\Wasp~\cite{zou2025wasp} replaces \PE's top-1 assignment with a
multi-candidate vote whose weights depend on candidate rank. This
changes how a record's voting mass is distributed among candidates,
but not how that mass responds continuously to the distance moved by
a protected substitution.

For a fixed candidate population, each candidate's rank is
piecewise constant as a function of the record location. A
substitution $x \rightarrow y$ that stays within the same ranking
region leaves the rank-based weights unchanged; crossing a rank
boundary changes one or more ranks discontinuously. This is the
same structural phenomenon that makes the top-1 nearest-neighbor map
piecewise constant over Voronoi cells
(Section~\ref{sec:theory-hardvote}), although a top-$Q$ rule can have
more than the two response levels of the one-hot vote.

Thus rank weighting can smooth how mass is shared \emph{across
candidates} without making the response graded in the
\emph{distance of the substitution}. Our construction acts on this
second axis: the softmax weights depend directly on candidate
distances, so moving a record by a short distance induces a
correspondingly bounded change in its vote. The two ideas are
compatible; a mechanism could combine multi-candidate selection with
a distance-sensitive weighting rule, which we leave to future work.

\paragraph{\Balldp{}'s radius restriction targets a different object.}
\Balldp~\cite{margaryan2026balldp} restricts standard DP's neighbor
relation to substitutions within a fixed radius of an underlying
distance. At the level of adjacency, this corresponds to a threshold
policy graph as in Definition~\ref{def:threshold-graph}. Its
application, however, is convex empirical risk minimization rather
than evolutionary synthesis.

Accordingly, the quantity that must be controlled is different.
\Balldp bounds how far the output of an optimization procedure can
move after a radius-bounded perturbation of one training record. That
change depends on the structure of the empirical-risk objective and
its optimization. Our setting instead releases, at each round, an
explicit aggregate of per-record votes. This makes the relevant
quantity the norm of the change in a fixed record-level response
$v_{C_t}(x)$, conditional on the current candidate population.
Neither the voting construction, the candidate evolution, nor the
closed-form soft-vote sensitivity bound studied here appears in
\Balldp.

The common point is the question posed by a radius-restricted neighbor
relation: whether restricting the allowed substitution distance
actually reduces the sensitivity of the object being released. The
answer depends on that object's response to the substitution. Our
hard-vote result provides an example where the restriction can fail to
give a graded reduction, while the soft vote is designed to exploit
it.

\paragraph{Generic softmax bounds versus the policy-substitution bound.}
Temperature-scaled softmax is a well-studied smooth approximation to
a hard selection rule. Generic Lipschitz bounds control
\[
\|\operatorname{softmax}(z)-\operatorname{softmax}(z+a)\|
\]
in terms of a norm of the arbitrary logit perturbation $a$; related
bounds arise both in work on the exponential
mechanism~\cite{mcsherry2007mechanism,optimal-approx-softmax} and in
direct analyses of the softmax map~\cite{nair2026softmax}.

For our logits
\[
z(x)_j=-d(x,c_j)/\tau,
\]
the triangle inequality gives
\[
|z(x)_j-z(y)_j| \le d(x,y)/\tau
\]
for every candidate $j$. A direct $\ell_2$ treatment therefore gives
\[
\|z(x)-z(y)\|_2
    \le \sqrt{m}\,d(x,y)/\tau,
\]
which introduces a dependence on the candidate count $m$ even though
the perturbation has additional structure.

The important structure is that softmax is invariant to adding the
same constant to all logits. Consequently, the absolute size of the
logit displacement is not the relevant quantity; only its range,
\[
\operatorname{osc}(a)=\max_j a_j-\min_j a_j,
\]
can change the normalized probabilities. For a policy substitution,
Section~\ref{sec:theory-softmax} shows
\[
\operatorname{osc}(z(x)-z(y))
    \le 2d(x,y)/\tau,
\]
with no dependence on $m$.

Lemma~\ref{lem:tv-tilt} then specializes the total-variation bound for
exponential tilts of Cohen and Fausti~\cite{cohen2023hyperbolic} to
obtain a nonlinear bound that saturates as the perturbation grows.
After conversion from $\ell_1$ to $\ell_2$, this gives
Corollary~\ref{cor:softmax},
\[
S_2(H_t,P)
\le
\sqrt{2}\tanh\!\left(\frac{r_d(G)}{2\tau}\right).
\]
Unlike a linear Lipschitz upper bound, this expression never exceeds
the $\sqrt{2}$ diameter relevant to two simplex vertices.
Proposition~\ref{prop:lower-bound} further shows that the bound is
attained for two candidates and can be approached arbitrarily closely
for any larger candidate population. Thus the contribution is not the
use of softmax itself, but the policy-specific oscillation argument
that yields a candidate-count-independent and tight calibration bound.

\paragraph{Constraint-aware synthesis and syntactic privacy.}
A policy graph is one way for a data owner to state application
knowledge that a private release should respect. PrivCI~\cite{ghahramanpour2026privci}
encodes a different kind of knowledge, conditional-independence (CI)
constraints between sensitive attributes and outcomes, inside a
marginal-based synthesizer under standard DP. It enforces these
constraints in the measurement step by adding feasibility checks to a
maximum-spanning-tree construction driven by the exponential mechanism,
so that the synthetic data does not reproduce spurious correlations.
PrivCI changes which statistics are measured while keeping the
protection target fixed. BF-Soft does the reverse: it keeps the
evolutionary measurement fixed and changes how the released statistic
responds to the substitutions the policy protects. The two directions
are orthogonal and could in principle be combined. Earlier work studies
privacy outside the DP setting. Diversity-constrained
anonymization~\cite{milani2021preserving,milani2023anonymization}
extends $k$-anonymization with owner-specified constraints on how often
selected attribute values must appear in the anonymized relation, and
PACAS~\cite{huang2018pacas,huang2020privacy} supports privacy-aware data
cleaning in which a service provider discloses generalized values from
its master data. Like Blowfish policies, these constraints let the data
owner state requirements explicitly. Unlike Blowfish, the protection is
syntactic and carries no noise-calibrated guarantee.

\section{PE's Post-Hoc Threshold \texorpdfstring{$H$}{H}}
\label{app:pe-threshold}

\PE's \texttt{DP\_NN\_HISTOGRAM}~\cite{lin2024privateevolution}
takes a threshold $H$ alongside the noise multiplier and distance
function. After Gaussian noise is added to the vote histogram,
entries of $\tilde H_t$ below $H$ are zeroed before resampling. The
purpose is practical: a hard-vote histogram over many candidates is
sparse, and without thresholding, candidates whose counts are
inflated primarily by noise may survive into
\texttt{ResampleMutate}. $H$ is tuned empirically in \PE and is not
adapted across rounds in~\cite{lin2024privateevolution}.

\paragraph{$H$ is post-processing.}
Every sensitivity result in Section~\ref{sec:theory} concerns the
pre-noise histogram $H_t(D;C_t)$. Its policy sensitivity
$S_2(H_t,P)$ determines the Gaussian noise scale required for a given
privacy budget. Thresholding is applied only after the noisy
histogram $\tilde H_t$ has been produced. It is therefore
post-processing: it consumes no additional privacy budget and changes
neither global nor policy-specific sensitivity. The same public
threshold can thus be applied after \TabPE, \Hard, or \GraphSmooth
without changing any privacy claim in Sections~\ref{sec:method}
and~\ref{sec:theory}.

\paragraph{Retuning $H$ under soft voting.}
Although thresholding does not affect privacy, its utility effect may
depend on the vote. A hard-vote histogram is typically sparse, whereas
a soft vote distributes each record's unit mass across several
candidates before noise is added. A threshold tuned for a sparse
hard-vote histogram may therefore remove legitimate soft-vote mass,
particularly at large $\tau$. At small $\tau$, the soft vote approaches
the hard vote and the original thresholding behavior is recovered.

Whether $H$ should be retuned specifically for soft voting is therefore
an implementation and utility question rather than a privacy question.
We do not analyze this choice here; all results in
Section~\ref{sec:theory} concern $H_t(D;C_t)$ before either noise or
post-processing.

\section{Algorithmic Comparison}
\label{app:algorithm}

Section~\ref{sec:method-instantiation} gives the main algorithmic
comparison: \TabPE, \Hard/\BFEvo, and \GraphSmooth differ only in
their per-record vote and sensitivity calibration.
This appendix records two details that are useful but not needed for
the main construction: the hard vote's behavior under other adjacency
notions and the cost of exact calibration on record-level policy
graphs.


\paragraph{Other adjacency notions.}
Corollary~\ref{cor:indicator} assumes bounded, unconstrained
replacement adjacency, under which one record's one-hot vote either
stays in the same bin or moves from one bin to another. The
corresponding $\ell_2$ changes are $0$ and $\sqrt2$.

With an add/remove edge $(u,\bot)$, one unit vote appears or
disappears, giving both $\ell_1$ and $\ell_2$ change $1$. More
generally, public constraints can make two minimally different
feasible databases differ in several rows. If at most $k$ aligned
record changes are required, triangle inequality gives the coarse
bounds $2k$ in $\ell_1$ and $k\sqrt2$ in $\ell_2$.
Table~\ref{tab:hard-vote} summarizes these cases. Our main analysis
and experiments use bounded, unconstrained replacement adjacency.

\paragraph{Exact calibration for Hard.}
For a fixed candidate population $C_t$, Hard's exact policy
sensitivity requires determining whether any protected edge
$(u,v)\in E(G)$ satisfies
\[
\mathrm{nn}_{C_t}(u)\ne \mathrm{nn}_{C_t}(v).
\]
A single crossing certifies the value $\sqrt2$, so the search can
terminate once one is found. To certify sensitivity $0$, however,
every protected edge must be checked. Because the population changes
across rounds, this test must in principle be repeated for each
$C_t$.

If nearest candidates are found by brute force, checking one edge
costs $O(m)$ candidate-distance evaluations per endpoint. A complete
scan therefore costs
\[
O\!\left(|E(G)|\cdot \mathrm{cost}(\mathrm{nn})\right)
\]
per round, where $\mathrm{cost}(\mathrm{nn})=O(m)$ for a brute-force
lookup and can be smaller when the metric space admits an appropriate
index. The universal value $\sqrt2$ is always a valid constant-time
fallback when exact policy calibration is not computed.

\paragraph{Record-level graph size.}
The relevant $E(G)$ above is the edge set of the
\emph{record-level} product graph, not merely the union of the
per-attribute value graphs. Suppose an edge protects the pair
$(a,a')$ on attribute $i$. At record level, that pair induces an edge
for every possible combination of the remaining attributes:
\[
(x_1,\ldots,a,\ldots,x_p)
\longleftrightarrow
(x_1,\ldots,a',\ldots,x_p).
\]
Consequently, even simple per-attribute policies can induce a
record-level graph whose number of edges grows exponentially with the
number of attributes. Exact \GraphSmooth sensitivity requires the
same maximum over this graph, because
Lemma~\ref{lem:policy-sensitivity} ranges over every protected
record-level edge.

This distinction is why the analytical \GraphSmooth calibration is
useful computationally. Once the weighted per-attribute reaches have
been obtained, Corollary~\ref{cor:attribute-decomp} gives the
record-level reach by taking their maximum, and
Corollary~\ref{cor:softmax} evaluates the resulting sensitivity bound
in constant time. No record-level graph needs to be enumerated.

\paragraph{Voting versus calibration cost.}
The voting step itself has the same asymptotic cost for all three
methods. For each private record, both nearest-neighbor assignment and
softmax voting require distances to the $m$ current candidates, hence
$O(m)$ work with a direct implementation. \GraphSmooth therefore
does not add asymptotic voting cost relative to \TabPE or \Hard.

The difference lies in calibration. \TabPE uses the fixed value
$\sqrt2$. \GraphSmooth uses a bound computed once from $r_d(G)$ and
$\tau$. Hard's exact policy calibration can instead require the
round-dependent graph search described above. The experiment in
Section~\ref{sec:exp-sensitivity-summary} compares these calibration
costs on a one-attribute domain; it does not measure end-to-end
synthesis time.

\section{Distances for Categorical Attributes}
\label{app:categorical}

Gower's default categorical term treats every pair of distinct values
alike: same value or not, with no nearer or farther mismatch. Under
this term, a protected categorical attribute contributes its full
Gower weight to the reach (Section~\ref{sec:method-categorical},
Remark~\ref{rem:participation}). The remedy lies in $d$, not in the
vote or the policy graph. If the attribute's values carry a distance
with more than one nonzero level, the soft vote and the bound of
Corollary~\ref{cor:softmax} apply unchanged, since neither assumes
$d$ is numeric, only that it is a metric.

When the values already fall into a taxonomy, such as job codes,
product categories, or administrative regions, a natural choice is a
\emph{hierarchy}. This is a rooted tree over the values, with
tree-path distance
\[
d(u,v)=\mathrm{depth}(u)+\mathrm{depth}(v)-2\,\mathrm{depth}(\mathrm{lca}(u,v)),
\]
the number of edges between $u$ and $v$ through their lowest common
ancestor. This is a metric. With all leaves at depth $L$, it takes the
nonzero values $2,4,\dots,2L$.

A threshold policy on this distance can therefore choose how deep a
protected substitution may reach, from siblings under one parent up
to the whole tree. After normalizing by the diameter $2L$, protecting
only sibling substitutions gives the attribute reach $w/L$ rather
than $w$. For \Bank's \texttt{job} ($L=2$), this halves the reach
from $0.2$ to $0.1$. That is still large compared with a narrow
numeric threshold such as \texttt{age}'s $0.017$, so a hierarchy
recovers gradation but not, at shallow depth, a small reach.

The flat block partition is the depth-$1$ case, with a single nonzero
distance. A binary attribute admits no gradation under any distance,
because a two-point domain has only one nonzero distance.

We use this construction only for \Bank's \texttt{job}, whose
grouping is given in Appendix~\ref{app:exp-policy}. Other categorical
attributes and the binary labels keep Gower's default term, since
they have no principled hierarchy to draw on.

\section{Full Proofs}
\label{app:proofs}

This appendix gives the complete proofs of the results stated in
Section~\ref{sec:theory}. The main text includes the central
derivations needed to understand the sensitivity bound; here we give
the remaining algebra and equality conditions in full.

\subsection{Proofs for Section~\ref{sec:theory-identity}}

\noindent\textbf{Lemma~\ref{lem:policy-sensitivity} (restated).}
For a linear query $f(D)=\sum_{x\in D}\phi(x)$, any policy graph
$G=(\mathcal{T},E)$, unconstrained $\mathcal{I}_Q$, and bounded
adjacency,
\[
S_p(f,P)
=
\max_{(u,v)\in E}
\|\phi(u)-\phi(v)\|_p,
\]
with value $0$ when $E=\emptyset$.

\begin{proof}
Under bounded, unconstrained adjacency, a policy neighbor $D'$ of
$D$ replaces exactly one record's value $u$ by some $v$ such that
$(u,v)\in E$, leaving every other record unchanged. Hence
\[
f(D)-f(D')
=
\sum_{x\in D}\phi(x)
-
\sum_{x\in D'}\phi(x)
=
\phi(u)-\phi(v).
\]
Conversely, every edge $(u,v)\in E$ can be realized by a pair of
policy-neighboring databases containing the same remaining records.
Taking the maximum over Definition~\ref{def:policy-sens} therefore
gives
\[
S_p(f,P)
=
\max_{(u,v)\in E}
\|\phi(u)-\phi(v)\|_p.
\]
If $E=\emptyset$, there are no protected substitution edges and the
maximum is defined as $0$.
\end{proof}

\noindent\textbf{Corollary~\ref{cor:attribute-decomp} (restated).}
For a product domain
$\mathcal{T}=\mathcal{T}_1\times\cdots\times\mathcal{T}_p$ with
per-attribute graphs $G_i=(\mathcal{T}_i,E_i)$ and record-level
product graph $G$,
\[
S_p(f,P)=\max_i S_{p,i}(f,P).
\]
Under a Gower-style distance
$d(x,y)=\sum_i w_i d_i(x_i,y_i)$,
\[
r_d(G)=\max_i w_i r_{d_i}(G_i).
\]

\begin{proof}
By construction, every edge of $G$ changes exactly one attribute.
Write
\[
E=\bigcup_{i=1}^p E_i',
\]
where $E_i'$ contains the record-level edges that agree on every
attribute except $i$ and differ on attribute $i$ by an edge of $E_i$.
These sets are disjoint. By Lemma~\ref{lem:policy-sensitivity},
\[
S_p(f,P)
=
\max_{(u,v)\in E}
\|\phi(u)-\phi(v)\|_p
=
\max_i
\max_{(u,v)\in E_i'}
\|\phi(u)-\phi(v)\|_p,
\]
which is the first claim.

For an edge $(u,v)\in E_i'$, all coordinates except $i$ agree, so
\[
d(u,v)
=
w_i d_i(u_i,v_i).
\]
Maximizing first over $E_i'$ and then over attributes gives
\[
r_d(G)
=
\max_i
w_i r_{d_i}(G_i).
\]
\end{proof}

Consequently, the analytical \GraphSmooth calibration depends only
on the largest weighted per-attribute reach. Exact soft-vote
sensitivity can still depend on the candidate population, and exact
Hard sensitivity depends on whether current assignments cross a
policy edge.

\subsection{Proof for Section~\ref{sec:theory-hardvote}}

\noindent\textbf{Corollary~\ref{cor:indicator} (restated).}
For the hard vote
$\phi(u)=e_{\mathrm{nn}(u)}$ and any policy graph
$G=(\mathcal{T},E)$,
\[
S_2(f,P)=\sqrt2
\]
if some $(u,v)\in E$ has
$\mathrm{nn}(u)\ne\mathrm{nn}(v)$, and
$S_2(f,P)=0$ otherwise.

\begin{proof}
If $\mathrm{nn}(u)=\mathrm{nn}(v)$, then
\[
\phi(u)=\phi(v)
\]
and the difference has norm $0$.

Otherwise,
\[
e_{\mathrm{nn}(u)}
-
e_{\mathrm{nn}(v)}
\]
has exactly one coordinate equal to $+1$, one equal to $-1$, and all
remaining coordinates equal to $0$. Its $\ell_2$ norm is therefore
$\sqrt2$. Lemma~\ref{lem:policy-sensitivity} takes the maximum of
this two-valued quantity over $E$.
\end{proof}

\subsection{Proofs for Section~\ref{sec:theory-softmax}}

\noindent\textbf{Lemma~\ref{lem:tv-tilt} (restated).}
Let $p$ be a probability distribution on
$\{1,\dots,m\}$ and
\[
q_j
=
\frac{p_j e^{a_j}}
{\sum_k p_k e^{a_k}}
\]
with $\mathrm{osc}(a)\le s$. Then
\[
\|p-q\|_1
\le
2\tanh(s/4),
\]
and the bound is attained by a suitable two-point distribution.

This is a specialization of the Hilbert-projective-metric bound of
Cohen and Fausti~\cite{cohen2023hyperbolic}. For completeness we give
the elementary argument.

\begin{proof}
Adding the same constant to all $a_j$ leaves $q$ unchanged. We may
therefore center them so that
\[
a_j\in[-s/2,s/2].
\]
Let
\[
W_j=e^{a_j},
\qquad
L=e^{-s/2},
\qquad
U=e^{s/2},
\]
so $W_j\in[L,U]$ and $LU=1$. Define
\[
Z=\sum_j p_jW_j
=
\mathbb{E}_p[W].
\]
Then
\[
q_j-p_j
=
p_j\frac{W_j-Z}{Z}.
\]
Because
\[
\sum_j p_j(W_j-Z)=0,
\]
the positive and negative parts have equal total mass, hence
\[
\|p-q\|_1
=
\frac{2}{Z}
\mathbb{E}_p[(W-Z)^+].
\]

For any $w\in[L,U]$, write
\[
w
=
\alpha(w)U+(1-\alpha(w))L,
\qquad
\alpha(w)
=
\frac{w-L}{U-L}.
\]
The function
\[
g_Z(w)=(w-Z)^+
\]
is convex. Since $L\le Z\le U$, its values at the endpoints are
$0$ and $U-Z$. Convexity therefore gives the chord bound
\[
(w-Z)^+
\le
\alpha(w)(U-Z).
\]
Taking expectation,
\[
\mathbb{E}_p[(W-Z)^+]
\le
\frac{Z-L}{U-L}(U-Z).
\]
Thus
\[
\|p-q\|_1
\le
\frac{2}{Z}
\frac{(Z-L)(U-Z)}{U-L}.
\]
Using $LU=1$, this becomes
\[
\|p-q\|_1
\le
\frac{2}{U-L}
\left[
(U+L)
-
\left(
Z+\frac1Z
\right)
\right].
\]
By AM--GM,
\[
Z+\frac1Z\ge2,
\]
with equality at $Z=1$. Since $1\in[L,U]$, this value is feasible and
maximizes the upper bound. Therefore
\[
\|p-q\|_1
\le
\frac{2(U+L-2)}{U-L}.
\]
Substituting
$U=e^{s/2}$ and $L=e^{-s/2}$,
\[
\frac{2(U+L-2)}{U-L}
=
\frac{2(\cosh(s/2)-1)}
{\sinh(s/2)}
=
2\tanh(s/4).
\]

For equality, the chord bound requires $W$ to be supported only on
$\{L,U\}$, and the second step requires $Z=1$. Choosing probabilities
so that
\[
\Pr(W=L)
=
\frac{U-1}{U-L},
\qquad
\Pr(W=U)
=
\frac{1-L}{U-L}
\]
satisfies both conditions. Hence a two-point distribution attains the
bound.
\end{proof}

\noindent\textbf{Corollary~\ref{cor:softmax} (restated).}
For the soft-vote histogram $H_t(D;C_t)$ and any policy graph
$G$,
\[
S_2(H_t,P)
\le
\sqrt2
\tanh\!\left(
\frac{r_d(G)}{2\tau}
\right),
\]
independent of $m=|C_t|$. For a threshold policy $G_\theta$, this
becomes
\[
S_2(H_t,P_\theta)
\le
\sqrt2\tanh(\theta/2\tau).
\]

\begin{proof}
By Lemma~\ref{lem:policy-sensitivity}, it suffices to consider one
policy edge $(x,y)\in E$. Let
\[
z(x)_j=-d(x,c_j)/\tau,
\qquad
z(y)_j=-d(y,c_j)/\tau.
\]
For each candidate,
\[
|z(x)_j-z(y)_j|
=
\frac{
|d(x,c_j)-d(y,c_j)|
}{\tau}
\le
\frac{d(x,y)}{\tau}
\]
by the reverse triangle inequality. Hence
\[
\mathrm{osc}(z(y)-z(x))
\le
\frac{2d(x,y)}{\tau}.
\]

Moreover,
\[
v_{C_t}(y)_j
=
\frac{
v_{C_t}(x)_j
\exp(z(y)_j-z(x)_j)
}{
\sum_k
v_{C_t}(x)_k
\exp(z(y)_k-z(x)_k)
},
\]
so $v_{C_t}(y)$ is precisely an exponential tilt of
$v_{C_t}(x)$. Applying Lemma~\ref{lem:tv-tilt} with
\[
s=\frac{2d(x,y)}{\tau}
\]
gives
\[
\|v_{C_t}(x)-v_{C_t}(y)\|_1
\le
2\tanh\!\left(
\frac{d(x,y)}{2\tau}
\right).
\]

Let
\[
z=v_{C_t}(x)-v_{C_t}(y).
\]
Because both votes sum to $1$,
$\sum_jz_j=0$. Writing $z=z^+-z^-$, the positive and negative
parts have equal $\ell_1$ mass, say $a$, so
\[
\|z\|_1=2a.
\]
Also,
\[
\|z\|_2^2
=
\|z^+\|_2^2+\|z^-\|_2^2
\le
\|z^+\|_1^2+\|z^-\|_1^2
=
2a^2
=
\frac{\|z\|_1^2}{2}.
\]
Therefore
\[
\|z\|_2
\le
\frac{\|z\|_1}{\sqrt2}.
\]
Combining the two bounds,
\[
\|v_{C_t}(x)-v_{C_t}(y)\|_2
\le
\sqrt2
\tanh\!\left(
\frac{d(x,y)}{2\tau}
\right).
\]
Finally,
$d(x,y)\le r_d(G)$ for every policy edge, so maximizing over
$E$ proves the claim.
\end{proof}

\noindent\textbf{Proposition~\ref{prop:lower-bound} (restated).}
For every $\tau>0$, $\theta>0$, and $m\ge2$, there exist a
finite domain $\mathcal{T}\subset\mathbb{R}$, $m$ candidates, and a
threshold-$\theta$ policy such that
\[
S_2(H_t,P_\theta)
=
\sqrt2\tanh(\theta/2\tau)
\]
for $m=2$, and for every $\varepsilon>0$,
\[
S_2(H_t,P_\theta)
\ge
\sqrt2\tanh(\theta/2\tau)-\varepsilon
\]
for $m>2$.

\begin{proof}
Fix $L\ge\theta$. Let
\[
c_1=0,
\qquad
c_2=L,
\]
and define
\[
u=\frac{L-\theta}{2},
\qquad
v=\frac{L+\theta}{2}.
\]
Then
\[
d(u,v)=\theta,
\]
so $(u,v)$ is a protected edge under a threshold-$\theta$ policy.
Also,
\[
d(u,c_2)-d(u,c_1)=\theta,
\qquad
d(v,c_2)-d(v,c_1)=-\theta.
\]

For $m=2$, the first coordinate of the soft vote is
\[
v_{C_t}(x)_1
=
\frac{
e^{-d(x,c_1)/\tau}
}{
e^{-d(x,c_1)/\tau}
+
e^{-d(x,c_2)/\tau}
}
=
\sigma\!\left(
\frac{
d(x,c_2)-d(x,c_1)
}{\tau}
\right),
\]
where $\sigma$ is the logistic function. Hence
\[
v_{C_t}(u)_1
=
\sigma(\theta/\tau),
\qquad
v_{C_t}(v)_1
=
\sigma(-\theta/\tau).
\]
The second coordinates are complementary, so
\[
\|v_{C_t}(u)-v_{C_t}(v)\|_2
=
\sqrt2
\left[
\sigma(\theta/\tau)
-
\sigma(-\theta/\tau)
\right].
\]
Using
\[
\sigma(z)-\sigma(-z)=\tanh(z/2),
\]
we obtain
\[
\|v_{C_t}(u)-v_{C_t}(v)\|_2
=
\sqrt2\tanh(\theta/2\tau).
\]
Corollary~\ref{cor:softmax} provides the matching upper bound, so
equality holds.

For $m>2$, place the remaining candidates sufficiently far from
$u$ and $v$. For example, take
\[
c_k=L+R+k,
\qquad
k=3,\dots,m.
\]
As $R\to\infty$,
\[
e^{-d(x,c_k)/\tau}\to0
\]
uniformly for $x\in\{u,v\}$ and $k\ge3$. The first two coordinates
of both votes therefore converge to the two-candidate softmax, while
the remaining coordinates converge to $0$. Consequently,
\[
\|v_{C_t}(u)-v_{C_t}(v)\|_2
\to
\sqrt2\tanh(\theta/2\tau).
\]
For any prescribed $\varepsilon>0$, some finite $R$ achieves a value
within $\varepsilon$, proving the claim.
\end{proof}

\subsection{Proof of Proposition~\ref{prop:temperature-tradeoff}}
\label{app:proof-tradeoff}

\noindent\textbf{Proposition~\ref{prop:temperature-tradeoff}
(restated).}
Fix $C_t=\{c_1,\dots,c_m\}$.
\begin{enumerate}[leftmargin=1.4em,label=(\alph*)]
\item For $x$ with margin
\[
\gamma(x)
=
\min_{j\ne\mathrm{nn}(x)}
d(x,c_j)
-
d(x,c_{\mathrm{nn}(x)})
>0,
\]
\[
1-v_{C_t}(x)[\mathrm{nn}(x)]
\le
(m-1)e^{-\gamma(x)/\tau}.
\]

\item With
\[
\mathrm{diam}(C_t)=\max_{j,k}d(c_j,c_k),
\]
for every $x$,
\[
\left\|
v_{C_t}(x)-\frac1m\mathds{1}
\right\|_1
\le
2\tanh\!\left(
\frac{\mathrm{diam}(C_t)}{4\tau}
\right)
\le
\frac{\mathrm{diam}(C_t)}{2\tau}.
\]
\end{enumerate}

\begin{proof}
For part (a), let
$c^*=c_{\mathrm{nn}(x)}$. Then
\[
v_{C_t}(x)[\mathrm{nn}(x)]
=
\frac{
1
}{
1+
\sum_{j\ne\mathrm{nn}(x)}
\exp\!\left(
-
\frac{
d(x,c_j)-d(x,c^*)
}{\tau}
\right)
}.
\]
By definition of $\gamma(x)$,
\[
d(x,c_j)-d(x,c^*)
\ge
\gamma(x)
\]
for every $j\ne\mathrm{nn}(x)$. Therefore
\[
v_{C_t}(x)[\mathrm{nn}(x)]
\ge
\frac{
1
}{
1+(m-1)e^{-\gamma(x)/\tau}
}.
\]
Hence
\[
1-v_{C_t}(x)[\mathrm{nn}(x)]
\le
\frac{
(m-1)e^{-\gamma(x)/\tau}
}{
1+(m-1)e^{-\gamma(x)/\tau}
}
\le
(m-1)e^{-\gamma(x)/\tau}.
\]

For part (b), let
\[
p=\frac1m\mathds{1},
\qquad
a_j=-d(x,c_j)/\tau.
\]
Tilting $p$ by $a$ gives exactly $q=v_{C_t}(x)$. Let $c_a$ and
$c_b$ be respectively the farthest and nearest candidates from $x$.
Then
\[
d(x,c_a)-d(x,c_b)
\le
d(c_a,c_b)
\le
\mathrm{diam}(C_t)
\]
by the triangle inequality. Thus
\[
\mathrm{osc}(a)
\le
\frac{\mathrm{diam}(C_t)}{\tau}.
\]
Lemma~\ref{lem:tv-tilt} gives
\[
\left\|
v_{C_t}(x)-\frac1m\mathds{1}
\right\|_1
\le
2\tanh\!\left(
\frac{\mathrm{diam}(C_t)}{4\tau}
\right).
\]
Finally, $\tanh z\le z$ for $z\ge0$, yielding
\[
\left\|
v_{C_t}(x)-\frac1m\mathds{1}
\right\|_1
\le
\frac{\mathrm{diam}(C_t)}{2\tau}.
\]
\end{proof}

\paragraph{Scope of Proposition~\ref{prop:temperature-tradeoff}.}
The proposition concerns one vote at one fixed candidate population.
It does not characterize the evolution of
$C_1,C_2,\dots$ under mutation and resampling, nor does it imply
convergence of the overall synthesizer. Those questions depend on the
full \texttt{ResampleMutate} dynamics. The sensitivity results of
Section~\ref{sec:theory}, in contrast, hold conditional on each fixed
population regardless of whether the population sequence converges.
The convergence analysis for the original hard-vote \PE has itself
required subsequent refinement~\cite{gonzalez2025peconverges}; we
therefore leave a corresponding analysis for soft voting to future
work.

\section{Additional Experimental Details and Results}
\label{app:experiments}

This appendix gives the policy rationale, experimental protocol, and
supporting results for Section~\ref{sec:experiments}.

\subsection{Policy Rationale}
\label{app:exp-policy}

We follow the style of Blowfish's motivating
examples~\cite{he2014blowfish}: protect a narrow range of
substitutions while allowing coarser distinctions to remain exposed.

For the numeric and ordinal attributes:
\begin{itemize}
\item The \texttt{age} and \texttt{hours.per.week} thresholds blur a
small number of years or hours without collapsing the full scale.
\item \texttt{education.num} on \Adult and \texttt{education} on
\Bank protect nearby schooling levels.
\item \texttt{balance}'s threshold of 500 protects balances that are
relatively close while leaving coarse differences, such as overdrawn
versus comfortably positive, distinguishable.
\end{itemize}

The binary labels \texttt{income} and \texttt{y} admit no nontrivial
intermediate graph: once their two distinct values are protected, the
attribute graph is complete. Protecting either therefore contributes
the attribute's full Gower weight to the record-level reach
(Remark~\ref{rem:participation}).

\paragraph{\texttt{workclass}: flat blocks.}
\texttt{workclass} on \Adult has nine values. We use a flat block
partition into private, government, self-employed, and other,
following the type of categorical grouping used in
Blowfish~\cite{he2014blowfish}. We do not add another hierarchy level.
For example, assigning a different policy distance to
\texttt{Federal-gov} and \texttt{Local-gov} would require an
additional privacy rationale not supplied by the attribute itself.
Thus \texttt{workclass} illustrates a categorical policy with coarse
blocks but no graded within-block distance.

\paragraph{\texttt{job}: hierarchical distance.}
\texttt{job} on \Bank has twelve values and supports a more structured
grouping. We use a two-level hierarchy with four broad labor-market
groups below the root:
\begin{itemize}
\item professional/technical:
\emph{management, technician, admin.};
\item manual/service:
\emph{blue-collar, services, housemaid};
\item independent:
\emph{self-employed, entrepreneur};
\item not in the labor force:
\emph{retired, student, unemployed, unknown}.
\end{itemize}
The twelve original categories form the leaves.

The \texttt{job} field is not coded using an external occupational
taxonomy, such as SOC, so this grouping is our construction rather
than an externally defined hierarchy. This is consistent with the
Blowfish model, in which the data owner specifies which distinctions
are protected.

We use the tree-path distance of \appref{app:categorical}. Two leaves
in the same parent group have distance $2$
(leaf--parent--leaf), while leaves in different groups have distance
$4$ (leaf--parent--root--parent--leaf). We normalize by the tree
diameter $4$ before using this distance in Gower distance.

A threshold policy on this tree has two nontrivial levels:
$\theta=2$ protects substitutions within the same broad group, while
$\theta=4$ protects every pair of job categories. This gives
\texttt{job} a genuine graded distance that the flat
\texttt{workclass} representation lacks.

\paragraph{\XORDS-3.}
The \XORDS-3 experiments protect $X_1$ only. Each input takes eight
ordinal values, $0,\ldots,7$. The default policy protects one-step
substitutions, while the temperature--reach experiments vary the
protected width from one to four steps. One step is $1/7$ of the
input's span, approximately $14.3\%$. With four equally weighted
modeled columns, it contributes record-level reach $1/28$.
The \XORDS-3 settings are therefore not defined by a common percentage
threshold and are not matched to the normalized reach or graph density
of \Adult or \Bank.

\subsection{Protocol and Reporting}
\label{app:exp-protocol}

\mypar{Implementation.}
All methods use the same implementation of
Algorithm~\ref{alg:unified}, equal-weight normalized Gower distance,
and hybrid selection. The standard experimental configuration uses
512 candidates in rounds 1--3 and 1{,}024 candidates in rounds 4--8.
Each run releases 512 records over eight rounds, with the privacy
budget divided equally across rounds. Random streams are paired across
methods.

\mypar{Repetitions.}
\begin{itemize}
\item Fixed-temperature budget sweeps use twelve paired runs at 17
budgets from $10^{-5}$ to $1$.
\item Temperature sweeps use six seeds at 20 temperatures from
$0.00025$ to $0.4$.
\item The smoothing-cost/noise-benefit decomposition uses twelve runs.
\item Matched-noise experiments use six seeds at 25 budgets.
\item \XORDS-3 temperature--reach experiments use six paired runs per
setting, cycling through three generated populations.
\item Public-pilot experiments use three pilot seeds and six paired
target seeds.
\item Broader-policy experiments use twelve paired runs.
\item The \XORDS-3 training-size experiment uses twelve independent
data/algorithm seed pairs, nested training prefixes, and a common
10{,}000-record test reference.
\end{itemize}

\mypar{Intervals and tests.}
Intervals are pointwise $95\%$ Student-$t$ intervals, not simultaneous
intervals. On \Adult and \Bank they quantify algorithmic randomness
for a fixed data split. For paired run $s$, relative error reduction is
\[
R_s
=
100\,
\frac{E_{H,s}-E_{S,s}}{E_{H,s}},
\]
and we report its mean $\bar R$. Thus we average the run-level
relative reductions rather than taking the ratio of the grand means.
Positive values favor \GraphSmooth.

Tests of zero mean paired reduction in
Section~\ref{sec:exp-policy} are two-sided. We apply Holm correction
over 680 prespecified primary contrasts. Results chosen retrospectively
as the best point in a parameter sweep are descriptive and are not
corrected for parameter selection.

\mypar{Sensitivity measurements.}
\begin{itemize}
\item
Figure~\ref{fig:exp-sensitivity-summary}\subref{fig:sens-temp}
uses a domain of eight binary attributes with single-bit substitution
edges. It covers four population sizes (32--256), three candidate
geometries, five seeds, and 22 temperatures including $\tau=0$, for
1{,}320 measurements. Each plotted curve is the median over the 60
configurations at a fixed temperature.

\item
Figure~\ref{fig:exp-sensitivity-summary}\subref{fig:sens-reach}
uses a 64-point ordinal domain, 20 candidates, six policy reaches, and
twelve candidate seeds. Bands show interquartile ranges.

\item
Figure~\ref{fig:exp-sensitivity-summary}\subref{fig:sens-runtime}
varies domain size and candidate count together over
$N\in\{64,128,\dots,8192\}$ with five paired seeds. Exact
computations evaluate the same $N\times N$ distance matrix and scan
the policy edges. These timings measure calibration only, not voting,
selection, mutation, or the full evolutionary loop.
\end{itemize}

\subsection{Fixed-Temperature Crossings}
\label{app:exp-crossings}

Table~\ref{tab:crossings} reports where the sampled mean errors of
\GraphSmooth and \Hard change order in
Figure~\ref{fig:exp-crossover}. Each value interpolates the mean
difference linearly in $\log\rho$ between two neighboring measured
budgets. No value extrapolates outside the experimental grid.
Consequently, these are descriptive between-sample crossing estimates,
not directly observed crossing budgets.


\subsection{Matched Noise Ratios}
\label{app:exp-equal-noise}

Figure~\ref{fig:exp-equal-noise} asks whether the noise ratio $q$
alone determines fidelity. Within each dataset, we hold $q$ fixed
while varying policy reach and temperature together.

For \Adult we use
\[
(\tau,\theta)
\in
\{(0.015,2),(0.03,4),(0.06,8)\}.
\]
For \Bank we set
\[
\tau
=
\frac{r_d(G)}
     {2\,\mathrm{atanh}(q)}
\]
at $\theta\in\{2,4,8\}$, giving
\[
(\tau,\theta)
\in
\{(0.01335,2),(0.02671,4),(0.05341,8)\}.
\]
Both datasets have $q=0.1806$.

For \XORDS-3 we use
\[
(\tau,\theta)
\in
\{(0.01,1),(0.02,2),(0.04,4)\},
\]
with reaches $1/28$, $2/28$, and $4/28$, respectively. All three have
$q=0.9453$. Thus $q$ is matched within the \XORDS-3 panel but not across
datasets. Each panel uses 25 budgets from $10^{-5}$ to $1$ and six
paired seeds.

The soft-vote curves differ despite equal noise ratios. Thus $q$
alone does not determine fidelity: temperature also changes the vote
and the subsequent evolutionary trajectory. Because the matched
settings protect different substitution ranges, this experiment
isolates equal-noise behavior rather than comparing mechanisms under
identical privacy policies.

\begin{figure*}[t]
 \centering
 \begin{subfigure}[b]{0.32\textwidth}
  \centering
  \includegraphics[width=\textwidth]{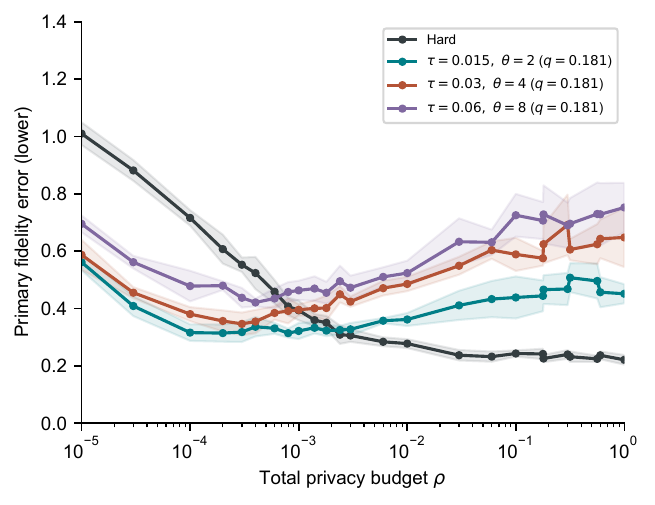}
  \caption{\Adult}
  \label{fig:eqnoise-adult}
 \end{subfigure}\hfill
 \begin{subfigure}[b]{0.32\textwidth}
  \centering
  \includegraphics[width=\textwidth]{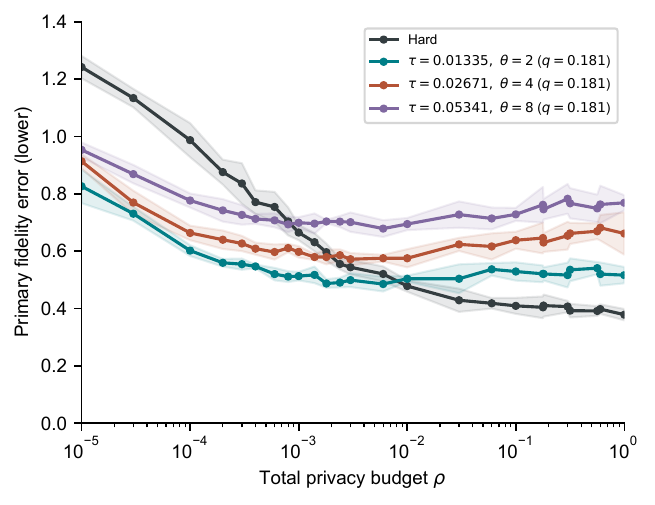}
  \caption{\Bank}
  \label{fig:eqnoise-bank}
 \end{subfigure}\hfill
 \begin{subfigure}[b]{0.32\textwidth}
  \centering
  \includegraphics[width=\textwidth]{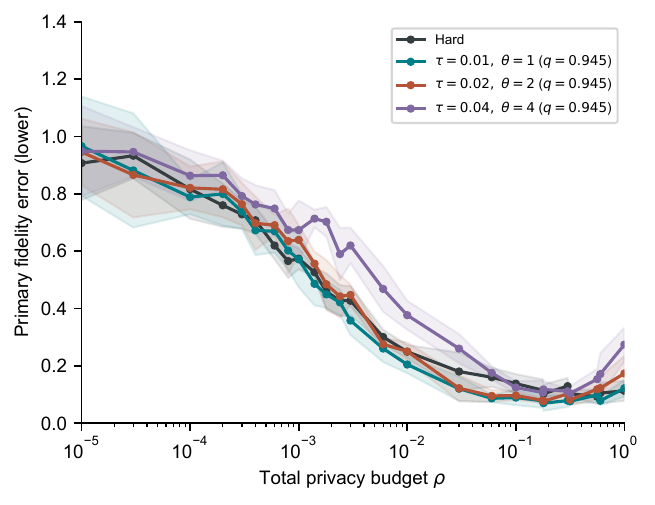}
  \caption{\XORDS-3}
  \label{fig:eqnoise-xor3}
 \end{subfigure}
 \caption{Error across privacy budgets at matched noise ratios.}
 \label{fig:exp-equal-noise}
\end{figure*}

\subsection{Temperature Sweeps on \Adult and \XORDS-3}
\label{app:exp-interior}

Figure~\ref{fig:exp-interior-full} gives the \Adult and \XORDS-3
temperature sweeps corresponding to the \Bank panels in
Figure~\ref{fig:exp-interior}. The numerical summaries in
Section~\ref{sec:exp-interior} use the lowest sampled
\GraphSmooth mean in each panel.

For \XORDS-3 at $\rho=3\times10^{-3}$, the best tested temperature is
$\tau=0.023$, with $q\approx0.65$. At $\rho=10^{-4}$, the absolute
improvement at the best tested setting is $0.0022$, illustrating that
a positive relative gain can still correspond to a small absolute
difference on this workload.

\begin{figure*}[t]
 \centering
 \begin{subfigure}[b]{0.32\textwidth}
  \centering
  \includegraphics[width=\textwidth]{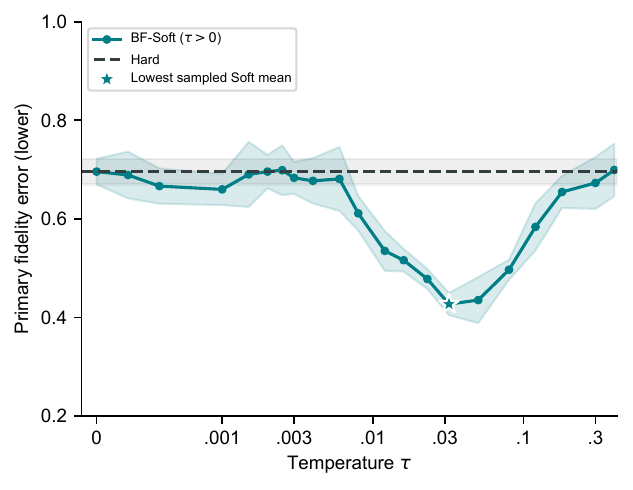}
  \caption{\Adult, $\rho=10^{-4}$}
  \label{fig:interior-adult-a}
 \end{subfigure}\hfill
 \begin{subfigure}[b]{0.32\textwidth}
  \centering
  \includegraphics[width=\textwidth]{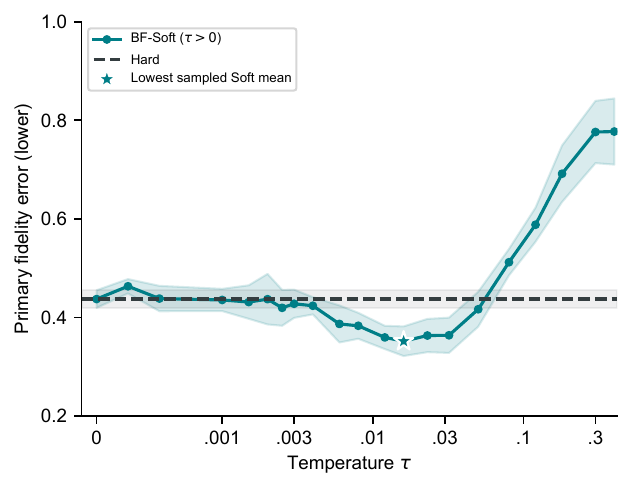}
  \caption{\Adult, $\rho=6\times10^{-4}$}
  \label{fig:interior-adult-b}
 \end{subfigure}\hfill
 \begin{subfigure}[b]{0.32\textwidth}
  \centering
  \includegraphics[width=\textwidth]{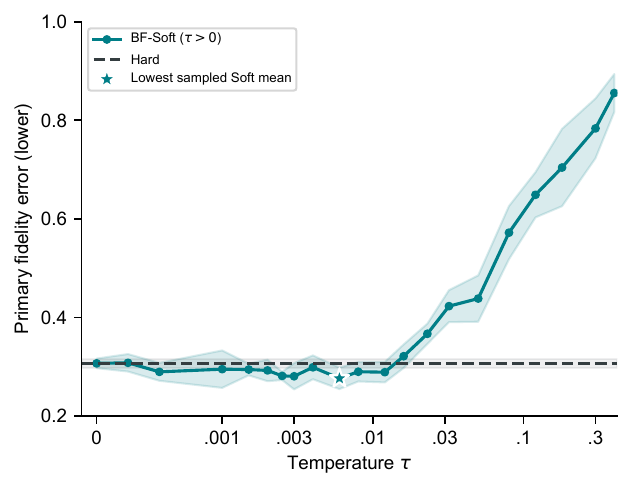}
  \caption{\Adult, $\rho=3\times10^{-3}$}
  \label{fig:interior-adult-c}
 \end{subfigure}

 \vspace{1em}

 \begin{subfigure}[b]{0.32\textwidth}
  \centering
  \includegraphics[width=\textwidth]{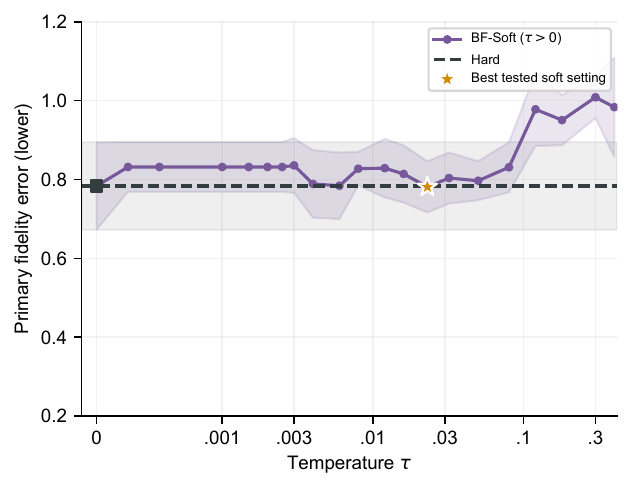}
  \caption{\XORDS-3, $\rho=10^{-4}$}
  \label{fig:interior-xor3-a}
 \end{subfigure}\hfill
 \begin{subfigure}[b]{0.32\textwidth}
  \centering
  \includegraphics[width=\textwidth]{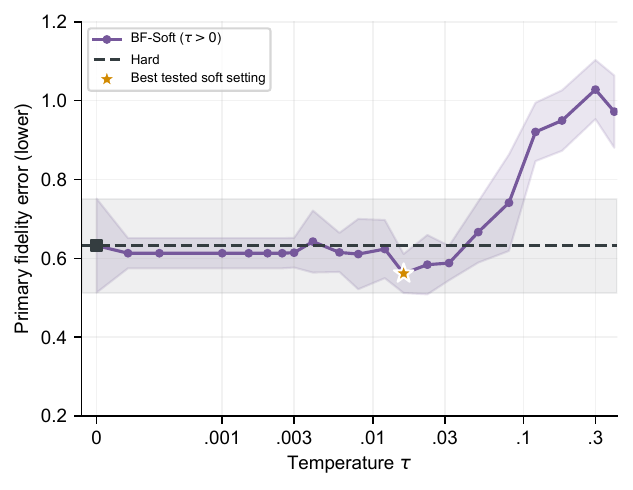}
  \caption{\XORDS-3, $\rho=6\times10^{-4}$}
  \label{fig:interior-xor3-b}
 \end{subfigure}\hfill
 \begin{subfigure}[b]{0.32\textwidth}
  \centering
  \includegraphics[width=\textwidth]{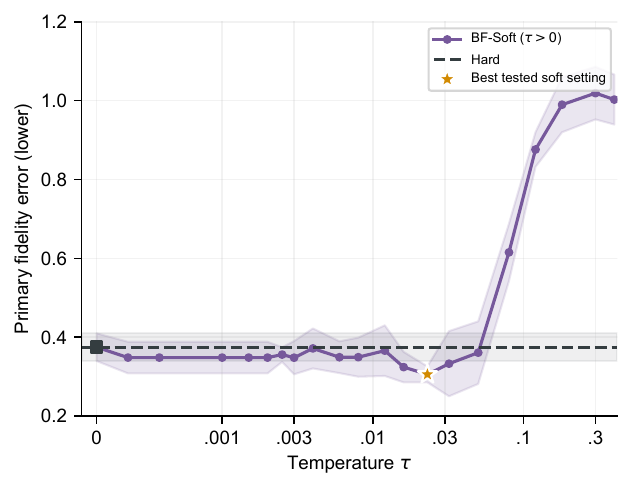}
  \caption{\XORDS-3, $\rho=3\times10^{-3}$}
  \label{fig:interior-xor3-c}
 \end{subfigure}

 \caption{Temperature sweeps on \Adult and \XORDS-3.}
 \label{fig:exp-interior-full}
\end{figure*}

\subsection{Temperature--Reach Interaction on \XORDS-3}
\label{app:exp-xor-interaction}

Figure~\ref{fig:exp-xor-interaction-full} jointly varies $\tau$ and
the protected width
$\theta\in\{1,2,3,4\}$ at three privacy budgets, with 48 evaluated
settings per panel.

The largest mean relative gains are $39.8\%$, $40.4\%$, and $19.9\%$
at $\rho=0.03$, $0.18$, and $1$, respectively
($\epsilon\approx1.2$, $3.1$, and $7.8$):
\begin{itemize}
\item at $\rho=0.03$, the largest measured mean gain occurs at
$\tau=0.003$ for the tied reaches
$\theta\in\{2,3,4\}$;
\item at $\rho=0.18$ and $\rho=1$, it occurs at
$(\tau,\theta)=(0.01,3)$.
\end{itemize}

At each of these settings,
\[
\frac{r_d(G)}{2\tau}>5,
\]
so $q\approx1$. The policy therefore saves essentially no noise.
The observed improvement instead comes from the effect of soft voting
on selection for this high-order workload. A soft vote calibrated to
standard DP would receive essentially the same noise at these
settings.

Only integer reaches from one to four steps were evaluated, and we did
not separately sample temperatures between the hard-vote limit
$\tau=0$ and the smallest positive tested temperature.

\begin{figure*}[t]
 \centering
 \begin{subfigure}[b]{0.32\textwidth}
  \centering
  \includegraphics[width=\textwidth]{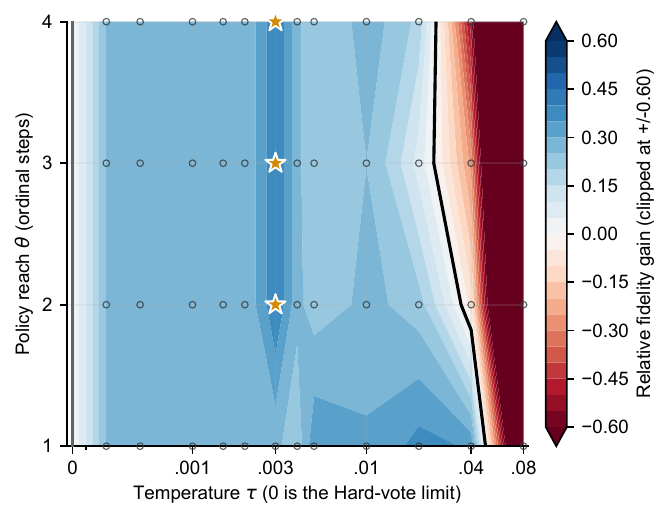}
  \caption{$\rho=0.03$}
  \label{fig:xorfull-a}
 \end{subfigure}\hfill
 \begin{subfigure}[b]{0.32\textwidth}
  \centering
  \includegraphics[width=\textwidth]{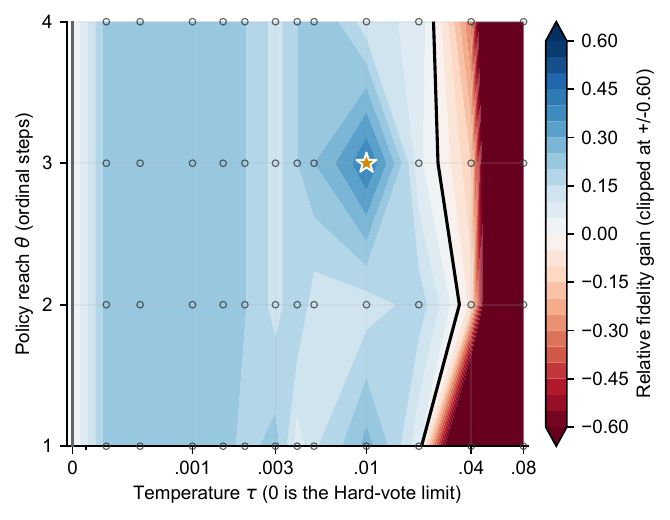}
  \caption{$\rho=0.18$}
  \label{fig:xorfull-b}
 \end{subfigure}\hfill
 \begin{subfigure}[b]{0.32\textwidth}
  \centering
  \includegraphics[width=\textwidth]{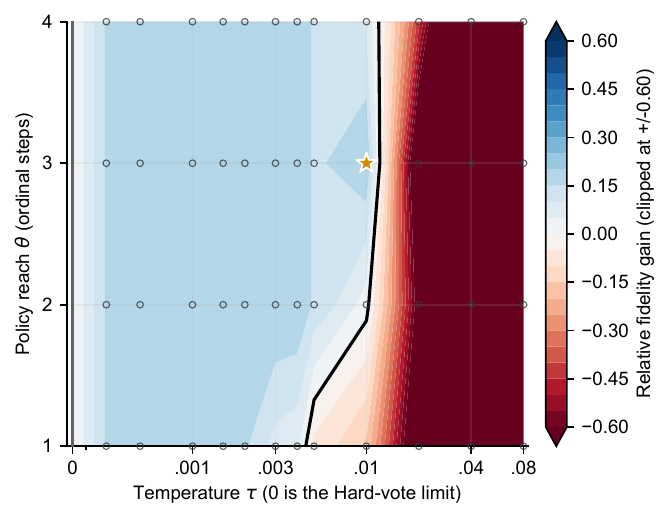}
  \caption{$\rho=1$}
  \label{fig:xorfull-c}
 \end{subfigure}
 \caption{\XORDS-3 relative gain over temperature and policy reach.
 Black contours interpolate zero mean gain.}
 \label{fig:exp-xor-interaction-full}
\end{figure*}

\subsection{Public-Pilot Forecasts}
\label{app:exp-forecast}

The public-data pilot in Section~\ref{sec:exp-tau-selection} uses
3{,}000 records disjoint from the private training sample. For
\Adult and \Bank these are held-out records from the same benchmark
distribution; \XORDS-3 uses independently generated public data.

The sign-accuracy result reported in
Section~\ref{sec:exp-tau-selection} uses the original 210-setting
\Adult grid. A later extension added lower temperatures to the plotted
grid in Figure~\ref{fig:exp-tau-selection}; those additional cells are
not included in the reported 190/210 sign-accuracy calculation.

The forecast agrees most closely with the observed target boundary at
the two smaller privacy budgets. A contour is absent when the mean gain
does not change sign on the evaluated grid. Where the forecast and
target contours disagree, the pilot predicts the wrong direction of
the mean difference.

The pilots come from the same benchmark distribution as the target.
These results therefore test whether a held-out public sample can
guide temperature selection under matched data generation, not whether
the procedure is robust to independently collected or shifted public
data.

\subsection{Broader Policies}
\label{app:exp-broader}
\label{app:extraExperiment}

The broader-policy experiments vary which of the selected attributes
are protected while keeping the modeled schema, Gower weights, and
evaluation workload fixed. A record-level policy edge still changes
only one coordinate. Protecting several attributes therefore means
taking the union of their single-attribute substitution edges; it does
not permit arbitrary simultaneous changes to multiple attributes.

These experiments concern the five modeled columns in
Table~\ref{tab:reference-sens}, not every column in the original UCI
datasets.

\paragraph{\Adult.}
At $\tau=0.03$ and $\rho=10^{-4}$, the \texttt{age} and
\texttt{hours.per.week} policies use thresholds of seven years and
eight hours. Their weighted reaches are $0.0192$ and $0.0163$,
respectively. Protecting both therefore leaves the record-level reach
at $0.0192$, the same as protecting \texttt{age} alone. Consequently,
the analytical calibration and paired trajectories are unchanged.

Adding attributes with larger reaches eventually changes this
behavior. Protecting all five modeled columns raises the reach to
$0.2$, increasing $q$ from $0.3092$ to $0.9975$ and leaving almost no
policy-specific noise reduction.

\paragraph{\Bank.}
For the broader-policy \Bank experiment, temperature is selected using
a public pilot before target evaluation. The pilot contains 3{,}000
validation records and uses a disjoint reference drawn from the
remaining validation records. Its budget is scaled according to
Eq.~\eqref{eq:exp-pilot-scaling}.

For each policy, the selector chooses among:
\begin{itemize}
\item \Hard;
\item
$\tau\in
\{0.001,0.003,0.015,0.03,0.06,0.12,0.3,0.6\}$;
\item the temperature giving $q=0.2$ for that policy.
\end{itemize}

Selection minimizes mean public error. Ties favor \Hard and then the
smaller temperature. No target errors are used during selection.

As discussed in Section~\ref{sec:exp-policy}, numeric-threshold
policies preserve substantial strong-privacy gains, while adding
attributes with larger categorical or binary reach reduces the
available noise saving. Pilot-selected gains under the broadest
policies do not survive the Holm correction and should be interpreted
as exploratory.

\subsection{Private Training Size}
\label{app:exp-trainsize}

Figure~\ref{fig:exp-trainsize} varies the \XORDS-3 private training size
over
\[
n\in\{3{,}000,12{,}000,48{,}000,100{,}000\}.
\]
The experiment keeps three inputs, 512 released records,
$\tau=0.03$, and one-step protection of $X_1$, giving
\[
r_d(G)=1/28,
\qquad
q=0.5337.
\]

At $\rho=0.003$, the mean paired error reduction is:
\begin{itemize}
\item $22.83\%$ ($[15.45,30.22]\%$) at $n=12{,}000$;
\item $37.37\%$ ($[25.19,49.55]\%$) at $n=48{,}000$;
\item $-89.18\%$ ($[-122.06,-56.30]\%$) at $n=100{,}000$.
\end{itemize}

Between $n=12{,}000$ and $n=100{,}000$, \Hard's mean error falls from
$0.389$ to $0.118$, while \GraphSmooth's falls from $0.297$ to
$0.214$. Both methods therefore improve with more private data, but
\Hard improves faster.

The privacy budget is fixed across training sizes. Because the
histogram noise scale does not grow proportionally with $n$, its size
relative to the empirical histogram shrinks as $n$ increases. Larger
$n$ therefore acts similarly to moving toward a weaker-noise regime,
where the cost of smoothing can dominate. This experiment concerns
one-step protection of $X_1$ at fixed parity order and does not imply
the same trend for every policy or workload.

\begin{figure}[t]
 \centering
 \includegraphics[width=\columnwidth]
 {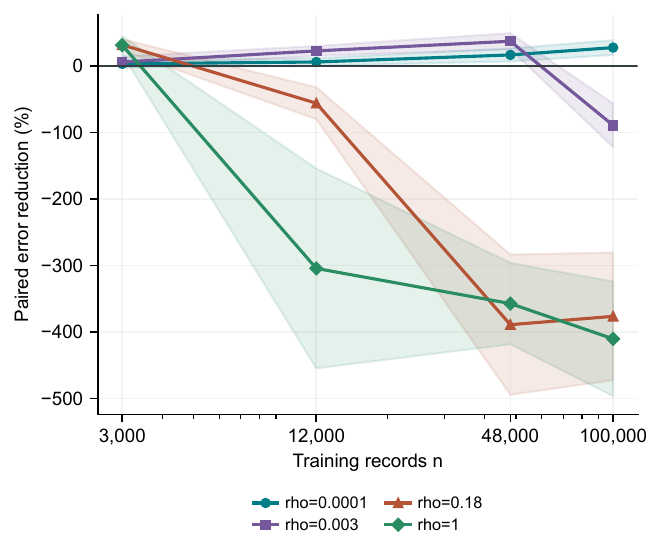}
 \caption{\XORDS-3 paired error reduction relative to \Hard across private
 training sizes.}
 \label{fig:exp-trainsize}
\end{figure}
\fi

\end{document}